\pdfoutput=1
\RequirePackage{fix-cm}
\documentclass[moor,nonblindrev]{informs3opt}
\usepackage{soul}
\usepackage{palatino,url}%
\usepackage{verbatim,amssymb}
\usepackage{amsmath,mathrsfs,thmtools,thm-restate}
\usepackage{graphics,amsopn,amstext,amsfonts,color}
\usepackage{bm}
\usepackage{setspace}
\usepackage{subfigure}
\usepackage{graphicx}
\usepackage{epstopdf}
\usepackage[sort&compress]{natbib}
\usepackage{float}
\usepackage{algpseudocode, algorithm}

\usepackage{multirow, multicol}
\usepackage{paralist}
\usepackage{colortbl}
\usepackage{titlesec}
\usepackage{makecell}
\usepackage[titletoc,toc,title]{appendix}
\usepackage{thm-restate}

\usepackage{booktabs}

\usepackage{xcolor}
\usepackage{cases}
\usepackage{empheq}
\usepackage[hidelinks]{hyperref}
\usepackage{graphics,latexsym,amsfonts,verbatim,amsmath}

\def\E{{\mathbb E}}

\def\Re{\mathbb{R}}

\def\hat{\widehat}

\def \A{\mathcal{P}}
\def \R{\mathcal{R}}

\def\Re{{\mathbb R}}

\newcommand{\ignore}[1]{{}}

\newcommand{\exclude}[1]{}
\newcommand{\bfx}{\bm{x}} 
\newcommand{\bfX}{\bm{X}} 
\newcommand{\bfp}{\bm{p}} 
\newcommand{\bfP}{\bm{P}} 
\newcommand{\bfv}{\bm{v}}

\newcommand{\bfz}{\bm{z}}

\makeatletter
\renewcommand*{\top}{%
  {\mathpalette\@transpose{}}%
}
\newcommand*{\@transpose}[2]{%
  \raisebox{\depth}{$\m@th#1\scriptscriptstyle\mathsf{T}$}%
}
\makeatother

\usepackage{authblk}
\usepackage{cleveref}

\declaretheorem[name=Proposition]{proposition}
\declaretheorem[name=Lemma]{lemma}

\newcommand*{\QEDA}{\hfill\ensuremath{\square}}
\newcommand*{\QEDB}{\hfill\ensuremath{\diamond}}

\usepackage{natbib}
\NatBibNumeric
\def\bibfont{\small}%
\def\bibsep{\smallskipamount}%
\def\bibhang{24pt}%

\bibpunct[, ]{[}{]}{,}{n}{}{,}%

\allowdisplaybreaks
\begin{document}

\TITLE{Computing Experiment-Constrained\\  D-Optimal Designs}
\RUNAUTHOR{A. Pillai, G. Ponte, M. Fampa, J. Lee,  M. Singh, W. Xie}
\RUNTITLE{Large-Scale  $D$-optimal Design
}

\ARTICLEAUTHORS{%
\AUTHOR{Aditya Pillai}
	\AFF{H. Milton Stewart School of Industrial and Systems Engineering, Georgia Institute of Technology, Atlanta, GA 30332, \EMAIL{apillai32@gatech.edu.}}
\AUTHOR{Gabriel Ponte}
\AFF{IOE Dept., University of Michigan and Federal University of Rio de Janeiro, \EMAIL{gabponte@umich.edu}}
\AUTHOR{Marcia Fampa}
\AFF{COPPE, Federal University of Rio de Janeiro, \EMAIL{fampa@cos.ufrj.br}}
\AUTHOR{Jon Lee}
\AFF{IOE Dept., University of Michigan, Ann Arbor, MI, \EMAIL{jonxlee@umich.edu}}
\AUTHOR{Mohit Singh}
	\AFF{H. Milton Stewart School of Industrial and Systems Engineering, Georgia Institute of Technology, Atlanta, GA 30332, \EMAIL{mohitsinghr@gmail.com.}}
\AUTHOR{Weijun Xie}
	\AFF{H. Milton Stewart School of Industrial and Systems Engineering, Georgia Institute of Technology, Atlanta, GA 30332, \EMAIL{wxie@gatech.edu.}}
}

\ABSTRACT{In optimal experimental design, the objective is to select a limited set of experiments that maximizes information about unknown model parameters based on factor levels. This work addresses the generalized D-optimal design problem, allowing for nonlinear relationships in factor levels. We develop scalable algorithms suitable for cases where the number of candidate experiments grows exponentially with the factor dimension, focusing on both first- and second-order models under design constraints. Particularly, our approach integrates convex relaxation with pricing-based local search techniques, which can provide upper bounds and performance guarantees. Unlike traditional local search methods, such as the ``Fedorov exchange" and its variants, our method effectively accommodates arbitrary side constraints in the design space. Furthermore, it yields both a feasible solution and an upper bound on the optimal value derived from the convex relaxation. Numerical results highlight the efficiency and scalability of our algorithms, demonstrating superior performance compared to the state-of-the-art commercial software, \texttt{JMP}.
}


\maketitle

\section{Introduction.}\label{sec:intro}
Optimal experimental design is a classical problem in statistics. The goal of optimal experimental design is to identify a limited number of experiments to obtain the maximum information about a vector of unknown model parameters in a given model based on levels of independent variables/factors.
In this work, and as is common, we assume that the model is linear in the parameters. However, the model is often nonlinear in the factor levels. 
We focus on one of the most popular optimal design criteria,  D-optimality. In D-optimal experimental design, under i.i.d. Gaussian noise, we aim to identify a set of experiments to
minimize the generalized covariance of the least-squares parameter estimator. In other words, we seek to minimize the volume of the standard confidence ellipsoid for the parameter centered on the least-squares parameter estimate. This is equivalent to maximizing the (logarithm of the) determinant of the Fisher information matrix (see \cite{Draper} for more details; also see \cite{Wald,Kiefer,KW,Puk,Fedorov} for history and more information).

We let integer $d\geq 1$ denote the dimension of factors and let $\bfx:=(x_1,\ldots, x_d)^\top$
represent a vector of factors.
In our most general setting,
we consider a model based on
a list of $p$ distinct monomials $\mathcal{M} := \{m_1(\bfx), \ldots m_p(\bfx)\}$ in the variables. Specifically, 
\[
y \approx 
\sum_{i=1}^p \theta_i m_i(\bfx),
\]
where ``$\approx$'' indicates that the model incorporates 
a zero-mean Gaussian noise. The model is
linear in the parameter vector $\bm\theta\in \Re^p$, and so, by choosing a set of
experiments $\bfX=\{\bfx^1,\ldots,\bfx^k\}$ and their associated observations, $y^1,\ldots, y^k$ respectively, we can calculate the least-squares estimate $\hat{\theta}=\argmin_{\theta} \sum_{j=1}^k \sum_{i=1}^p \left(\theta_i m_i(\bfx^j) -y^j\right)^2$. 
The challenge, of course, is to choose the best experiments, $\bfX=\{\bfx^1,\ldots,\bfx^k\}$. 
Following the convention, we assume that experiments can be repeated.

For there to be a unique least-squares solution to the linear regression problem associated
with a set of experiments (i.e., the model is ``identifiable''
given the experiments),
we need a set of design points of rank $p$. Assuming that the model is identifiable, given a
multi-set of experiments $\bfX $, 
the D-optimality criterion
gives the value 
$\log  \det( \sum_{\bfx\in\bfX } \bfp(
\bfx)  \bfp(\bfx)^\top )$,
which we wish to maximize.

Before continuing, it is useful to settle on some terminology.
If the maximum degree of a monomial in $\mathcal{M}$ is $\ell$, then we say that we have an \emph{$\ell$-th order model}. 
In an $\ell$-th order model, if we have all possible monomials of maximum degree $\ell$, then we say that it is a \emph{full $\ell$-th order model}. 
We refer to models that are not full as \emph{partial}. 
In what follows, we will be mostly interested in
full first-order models (having $d+1$ monomials), full second-order models (having $\binom{d}{2}+d+1$ monomials), and partial second-order models.  

The \emph{design point} associated
with an experiment $\bfx\in\mathbb{R}^d$ is 
simply 
$\bfp(\bfx):=(m_1(\bfx),\ldots,
m_p(\bfx))^\top\in \mathbb{Z}^p$. 
That is, it is the vector of evaluations of each monomial on the experiment $\bfx$.
We assume that we have 
a set of $L\geq 2$ levels $\{0, 1,\ldots, L-1\}$, that are the possible values of the factors.
So, without additional structural
requirements, we have a universe of $L^d$ possible experiments, one for each $\bfx\in \{0,1\,\ldots, L-1\}^d$.
It is common for one to have $L\geq \ell+1$
for a full $\ell$-th
order model. For example, to capture
a second-order effect 
at least in one variable, we would need at least three levels
(this is the motivation for many second-order designs in the literature 
that use three levels, such as the well-known Box-Behnken design; see \cite[Sec. 5.3.3.6.2]{BoxBehnkenNIST} and \cite{BB1960}).


We note that modeling and algorithms associated with constrained design regions are a
central topic in the design of experiments literature; see \cite{DOE1,DOE2,DOE3} for references. 
Specifically, we consider 
the constraints of form 
$\bm{A}\bfx \leq \bm{b}$, for $\bfx\in \mathcal{Y} $. So, we define our universe of allowable experiments as 
\[
\mathcal{Y}:=\{ \bfx\in \{0,1,\ldots,L-1\}^d ~:~ \bm{A}\bfx \leq \bm{b}\}.
\]
Finally, 
we impose a ``budget restriction''  on the size of $\bfX$,
the chosen multi-set of experiments, to be $k$.

With all of this, we are now ready to formally state our \emph{experiment-constrained D-Optimal design problem}:
\begin{align} \label{CD-Opt}\tag{ECD-Opt}
\phi:= \max \left\{ \log \det\left( \sum_{\bfx\in \mathcal{Y}} \lambda(\bfx)
\bfp(
\bfx)  \bfp(\bfx)^\top \right)
~:~ \sum_{\bfx\in\mathcal{Y}} \lambda(\bfx)=k;~ \lambda(\bfx)\in\mathbb{Z}_+\,, \mbox{ for } \bfx\in \mathcal{Y}
\right\}.  
\end{align}
The solution selects the multi-set of experiments $\bfX$, by simply having $\lambda(\bfx)$ copies of $\bfx$ in $\bfX$.
We can further associate matrices  
$\bfP$ and $\bm{S}$ where $\bfP$ has 
a row $\bfp(\bfx)$ for
each $\bfx\in \bfX $ and matrix $\bm{S}=\bfP^\top \bfP$.
In this way,
we have 
$\bm{S}=\bfP^\top\bfP=\sum_{x\in \bfX} 
\bfp(
\bfx)  \bfp(\bfx)^\top$,
and we can rewrite the objective function as $\log \det(\bm{S})=\log \det(\bfP^\top\bfP)$.


In what follows,
we let
 the collection of feasible design points be
 \[
 \A:=\{\bfp(\bfx) \in \Re^p ~:~ \bfx \in \mathcal{Y}\}.
 \]
Therefore, \ref{CD-Opt} is equivalent to
\begin{align} \label{CD-Opt2}
\max \left\{ \log \det\left( \sum_{v\in \A} \lambda(\bfv)
\bfv  \bfv^\top \right)
~:~ \sum_{\bfv\in\A} \lambda(\bfv)=k;~ \lambda(\bfv)\in\mathbb{Z}_+\,, \mbox{ for } \bfv\in \A
\right\}.  
\end{align}
The main computational challenge is that the set of vectors $\A$  consists of exponentially many vectors in the dimension $p$, and in practice, may not be explicitly defined.

As studied in the literature and available software, when the set
of allowable experiments $\mathcal{Y}$ is given
explicitly, as
a list, it is interesting to have linear constraints on the integer variable vector $\bm\lambda$, whereupon 
the associated D-optimal design problem has been referred to as CD-Opt (see \cite{GMESP_arXiv}). In this paper, set $\mathcal{Y}$
is given by a set of
allowable experiment vectors $\bfx$, described via a range of levels for 
each factor
with optionally additional
linear constraints. 
To distinguish it from CD-Opt, we
call our problem \ref{CD-Opt}.

There is a great deal of literature on heuristic
algorithms for
\ref{CD-Opt}, mainly focusing on the case where there are no side constraints $\bm{A}\bfx \leq \bm{b}$, based on the usual greedy and local search ideas;  
see, for example, the references in \cite[Sec. 1]{KoLeeWayne}. 
The popular software \texttt{JMP}
embraces local search ideas
described in \cite{meyer1995coordinate}, see
\cite[Page 125]{jmp2010design} and is capable of handling additional constraints as well as settings where design points are exponential in size. In Section~\ref{sec:experiments}, we give a detailed comparison of our approach to \texttt{JMP}. We remark that \texttt{JMP} does not produce any gaps to optimality and focuses only on heuristic based approaches, in contrast to our setting.
Various branch-and-bound (B{\&}B) based approaches aiming to find exact optimal solutions  
were also investigated in \cite{KoLeeWayne,KoLeeWayne2,Welch,PonteFampaLeeSBPO22,PonteFampaLeeMPB}.  
None of these handle the case in which the 
set of 
feasible design points is very large and not explicitly 
given.
Our work fills this gap and aims to develop
scalable computational techniques and fast approximation algorithms with performance guarantees,
and to solve difficult instances of \ref{CD-Opt}.

\subsection{Our results and contributions.}
In this work, we consider algorithms for cases of \ref{CD-Opt} for which the number of candidate experiments is exponential in the number of factors, $d$. We develop algorithms for \ref{CD-Opt} and study its continuous convex relaxation. For the sake of practicality, we focus on the cases of full first-order models and second-order models. We would like to emphasize that, in contrast to elementary/classical local search methods (known as ``Fedorov exchange'' and its variants), our approach effectively handles arbitrary side constraints $\bm{A}\bfx \leq \bm{b}$. Moreover, our approach provides not only a solution for the problem but also an upper bound for the optimal value based on the convex relaxation of the problem.

\medskip

\noindent \textbf{Pricing-based local search.}
Consider a given feasible solution $\lambda(\bfx)$, for each $\bfx \in \mathcal{Y}$ of 
\ref{CD-Opt}.
We can view it as a multi-set of experiments
$\bfX $, where $\bfx \in \mathcal{Y}$ occurs with multiplicity $\lambda(\bfx)$ in $\bfX $.
A local search step, in our context, exchanges one experiment
$\bfx'\in \bfX $ with an experiment $\bfx \in \mathcal{Y}\setminus \{\bfx'\}$,
in a way that the objective of \ref{CD-Opt} increases. 
The classical ``Fedorov exchange method'' (i.e., local search step), typically applied when there are no or only very simple side constraints, treats 
set $\mathcal{Y}$ as an input list.
To relieve the computational burden, such a local search would
only consider $\bfx \in \bfX\setminus \{\bfx'\}$ such that $|\bfx-\bfx'|$ is a standard unit vector. In contrast, we treat set
$\mathcal{Y}$ as given by constraints, and we formulate the ``pricing problem for local search'', which seeks an improving $\bfx$ using 
mathematical optimization techniques. This is akin to the well-known column-generation method for a wide variety of problems (which is 
why we call it a ``pricing problem'', 
even though we do not use ``prices'' in the sense of dual variables).
Our pricing problem effectively searches over a much larger local neighborhood than the simple local search described above, since it has no limitation on $|\bfx-\bfx'|$.
For the case of an $\ell$-th order model, the pricing problem has a polynomial objective function of degree $2\ell$,
linear constraints, and integer variables. We will see that this is quite tractable for first-order models and for partial second-order models. 

With respect to first-order models, the pricing problem for our local search becomes a 
 maximization program with a convex quadratic objective function. 
This situation
can be directly handled with, for example, \texttt{Gurobi}; 
or, using standard lifting methods, we
can reformulate it linearly in the case where $L=2$.
If $L>2$, such a linear reformulation is still possible (by the so-called ``binarization''). See \cite{LEE2007} for the use of similar techniques for a completely different problem. 
For second-order models, we  
proceed similarly but with a greater computational burden.

\medskip
\noindent \textbf{Pricing for the continuous relaxation.}
Finally, we also study the continuous relaxation of \ref{CD-Opt} obtained after removing the integrality constraints,
which is a convex program.
In connection with this, we let the optimal value of that be $\phi_R$.
Solving the continuous relaxation 
provides a useful upper bound to indicate the quality of any feasible integer solution to 
\ref{CD-Opt}. Our techniques are very similar to those for the local search pricing problem. The main difference is that the objective function for the continuous relaxation pricing problem has a different polynomial objective function of degree $2\ell$.
We will refer
then to the ``pricing problem for the continuous relaxation''
to distinguish it from the ``pricing problem for local search''.
The other differences between the two pricing approaches, which are rather significant, are in the computational details. 
In particular,
we have found computational advantages by: (i) including randomly-generated experiments besides the best ones found by solving
the pricing problem, (ii) working with the dual of the continuous relaxation if necessary, and (iii) reducing the number of experiments via a sparsification procedure over the iterations.

\medskip

\noindent \textbf{Theoretical results.} We also theoretically analyze the performance of our pricing-based local search procedure. The polynomial optimization problem encountered at every local search step is, in general, NP-hard, but good approximation algorithms can be developed under certain settings. Building on the results~\cite{madan2019combinatorial}, we provide an approximation algorithm at each local search step to obtain an approximate local search procedure with the following performance guarantee.
%

 
\begin{restatable}{theorem}{thmlocal} \label{theorem:approx-local}
Suppose that there exists a polynomial-time $\frac{1}{\rho}$-approximation algorithm for some $\rho>1$ to the problem $\max_{ \bfv \in \A} \bfv^\top \bm{G} \bfv$ for every positive semidefinite matrix $\bm{G} \succeq 0$.
Then the pricing-based local search  
algorithm presented in Section~\ref{sec:localsearch} returns a 
feasible solution $\lambda(\bfx)$ of 
\ref{CD-Opt} with the associated 
 multi-set of experiments
$\bfX $, where $\bfx \in \mathcal{Y}$ occurs with multiplicity $\lambda(\bfx)$ in $\bfX $, and 
 associated 
matrix $\bm{S}$ 
satisfying 
\[ 
\det( \bm{S} ) \geq e^{\phi} \left(  \frac{k - p + 1}{k} \frac{p}{p + k (\rho - 1)} \right)^p
\]
where $\phi$ is the optimal objective of the \ref{CD-Opt}.
\end{restatable}

An interesting case where there exists an approximation algorithm for the sub-problem $\max_{ \bfv \in \A} \bfv^\top \bm{G} \bfv$ for every $\bm{G} \succeq 0$ is when 
$\A =\{ \bfx \in \{0, 1\}^d ~:~ x_1 = 1 \}$. This corresponds to the full first-order model with two levels, where the first monomial is the constant, i.e., $m_1(\bfx):=1$. 
Indeed, Nesterov \cite{nesterov2011random}, gives a $\rho = \frac{\pi}{2}$ approximation algorithm for the subproblem $\max_{ \bfv \in \A} \bfv^\top G \bfv$. Applying Nesterov's results along with Theorem~\ref{theorem:approx-local}, we obtain the following corollary. 
\begin{restatable}{corollary}{corrNest}\label{corollary:full-linear-guarantee}
 For the full first-order model with two levels and no side constraints (i.e., $\A=\mathcal{Y}$), 
 the pricing-based local search algorithm returns an approximate local optimum $\bfX$
 to \ref{CD-Opt} with associated matrix $\bm{S}$ satisfying 
 \[\det(\bm{S}) \geq e^{\phi}  \left(  \frac{k - p + 1}{k} \frac{p}{p + k (\frac{\pi}{2} - 1)} \right)^p \]
 and additionally, satisfying
 \[\det(\bm{S}) \geq \frac{k^p}{2^{2(p-1)}} \left(  \frac{k - p + 1}{k} \frac{p}{p + k (\frac{\pi}{2} - 1)} \right)^p. \]
 \end{restatable}
Here, the equality for $\phi_{R}$ follows using the symmetry of the convex relaxation.

\subsection{Organization.}
In Section~\ref{sec:localsearch}, we introduce the pricing based local search algorithm and discuss the ideas used in its implementation. Next, in Section~\ref{sec:local-theory}, we analyze the local search theoretically by proving an approximation guarantee given an approximation algorithm for the pricing problem. We also give an efficient approximation algorithm for pricing problem in the first-order model and also show a hardness result for the pricing problem in the second-order model. In Section~\ref{sec:relaxation}, we discuss the convex relaxation and how we implement a pricing based column generation to solve the relaxation. Next, in Section~\ref{sec:experiments}, we discuss numerical results of local search and column generation. 


\section{Local search algorithm for D-Optimal design.}\label{sec:localsearch}
In this section, we discuss the local search algorithm for \ref{CD-Opt}. 
We are interested in instances of the problem where the set of vectors  $\A:=\{\bfp(\bfx) \in \Re^p ~:~ \bfx \in \mathcal{Y}\}$ is not explicitly given and is too big to enumerate all its elements. We  consider that we have a feasible solution      $\hat{\lambda}(\bfx)$,  for   $\bfx\in \mathcal{Y}$, to \ref{CD-Opt},  but we only store the vectors   $\bfp(\bfx)$ for which $\hat{\lambda}(\bfx)>0$. We let $\mathcal{Y}^{>0}:=\{\bfx\in\mathcal{Y}~:~\hat{\lambda}(\bfx)>0\}$ and therefore, the associated matrix $\bm{S}=\sum_{\bfx\in \mathcal{Y}^{>0}}\hat{\lambda}(\bfx) \bfp(\bfx) \bfp(\bfx)^\top $, and we assume that rank$(\bm{S})=p$. Observe that $|\mathcal{Y}^{>0}|\leq p$ and thus can be explicitly stored. The simplest local search move decreases $\hat{\lambda}(\bfx')$ by one, for some $\bfx'\in\mathcal{Y}^{>0}$, 
and increases  $\hat{\lambda}(\bfx)$ by one, for some $\bfx\in\mathcal{Y}$ and  $\bfx\neq\bfx'$. We can explicitly iterate over the experiments in the set $\mathcal{Y}^{>0}$ to find $\bfx'$. 
Since we do not explicitly generate and store all the vectors of $\A$, for a given $\bfx'\in\mathcal{Y}^{>0}$, we will seek to select an experiment $\bfx\in\mathcal{Y}$ and generate the corresponding vector $\bfp(\bfx)$ of $\A$, so that $\det(\bm{S}-\bfp(\bfx') {\bfp(\bfx')}^\top + \bfp(\bfx) \bfp(\bfx)^\top)> \det(\bm{S})$, if one exists. 
 
We define $\bm{S}_{\bfx'} := \bm{S} - \bfp(\bfx'){\bfp(\bfx')}^\top$ for all $\bfx'\in\mathcal{Y}^{>0}$ and, for a $p \times p$  matrix $\bm{S}$ and $k<p$, we define $\overset{k}\det(\bm{S})$  as the product of the $k$ largest eigenvalues of $\bm{S}$. 
Then, for all $\bfx'\in\mathcal{Y}^{>0}$ and $\bfx\in\mathcal{Y}$, if  $\mbox{rank}(\bm{S}_{\bfx'}) =p$, we have
\[
  \det( \bm{S}_{\bfx'} + \bfp(\bfx){\bfp(\bfx)}^\top ) = \det( \bm{S}_{\bfx'})(1 + {\bfp(\bfx)}^\top \bm{S}_{\bfx'}^{-1}\bfp(\bfx)).
\]
Otherwise, if $\bm{S}_{\bfx'}$ is rank deficient, we have   
\[  \det( \bm{S}_{\bfx'} + \bfp(\bfx){\bfp(\bfx)}^\top ) =  \overset{p-1}{\det}(\bm{S}_{\bfx'}){\bfp(\bfx)}^\top (\bm{I}_p - \bm{S}_{\bfx'}^{\dag}\bm{S}_{\bfx'} ) \bfp(\bfx), \]
where $\bm{S}_{\bfx'}^{\dag}$ is the Moore-Penrose pseudoinverse of $\bm{S}_{\bfx'}$ ( see \cite{Penrose}).

Finally,  for a given $\bfx'\in\mathcal{Y}^{>0}$, generating $\bfp(\bfx)\in \A$ that  maximizes $\det( \bm{S}_{\bfx'} + \bfp(\bfx){\bfp(\bfx)}^\top )$ is equivalent to the  following pricing problem: 
\begin{equation}\label{pricing_LS} \tag{Sub}\max_{\bfx \in \mathcal{Y}}\bfp(\bfx)^\top\bm{G}\bfp(\bfx), \end{equation} where $
\bm{G}:= \bm{S}_{\bfx'}^{-1}$, if $\mbox{rank}(\bm{S}_{\bfx'}) =p$, and  $\bm{G}:= \bm{I}_p-\bm{S}_{\bfx'}^{\dag}\bm{S}_{\bfx'}$, otherwise.

It is easy to see (considering the $\bfx$ associated with $\bfx'$) that, if $\mbox{rank}(\bm{S}_{\bfx'}) =p$, the optimal objective value of \eqref{pricing_LS}
is at least $\det(\bm{S})/\det(\bm{S}_{\bfx'})-1$, and any objective value above that gives a local move with improving objective-function
value. Otherwise, 
the optimal objective value of \eqref{pricing_LS}
is at least  $\det( \bm{S}) /\overset{p-1}{\det}(\bm{S}_{\bfx'}) - 1$.

In Section~\ref{sec:Pricing}, we discuss how to solve the problem \eqref{pricing_LS}. 

\subsection{Solving the pricing problem for the local search.}\label{sec:Pricing}

For the case of an $\ell$-th order model,  the pricing problem \eqref{pricing_LS} maximizes a polynomial objective function of degree $2\ell$, subject to linear constraints and integrality constraints over the variables. For any order, this is a nonconvex
optimization problem, and a common approach to solving it consists of reformulating it
as an integer linear program by lifting the problem to a higher-dimensional space using McCormick inequalities. In the basic approach, we would instantiate
a new variable $y_i$ for each monomial $m_i(\bfx)$, and we
describe or approximate the convex hull of the graph of 
$y_i=m_i(\bfx)$  on the domain
$\{0,1,\ldots,L-1\}^d$,
using linear inequalities.
This can be done quite effectively.
For example, when $L=2$, we can assume that the degree of each variable in each monomial is 0 or 1. Therefore, we can define that $m_i(\bfx)=
\prod_{j\in \mathcal{I}} x_j$. Then
the inequalities 
$0\leq y_i \leq x_j$ for $j\in \mathcal{I}$
 and
$y_i \geq \sum_{j\in \mathcal{I}}
x_j -|\mathcal{I}| +1$ define the convex hull of the graph of 
$y_i=m_i(\bfx)$.
We note that one way to handle more than two levels is through ``binarization,''
representing each integer variable $0\leq x_i\leq L-1$ with a linear function of $\lceil \log_2 L\rceil$ binary variables (e.g., see \cite{Owen,Buchheim2008}). 

Although our approach can be applied to solve \eqref{pricing_LS} for any $\ell$-th order model, we are particularly interested in solving two special cases-- the full first- and second-order models.

For a full first-order model with $d$ factors, a design point is given by $\bfp(\bfx)=(1,x_1,\ldots,x_d)^\top \in\mathbb{Z}^{d+1}$. This particularizes \eqref{pricing_LS} 
to the quadratic optimization problem with integer variables $\alpha_0,\ldots, \alpha_d$:
\begin{equation}\tag{Quad\_Sub}\label{quad_sub}
\begin{array}{lrl}\displaystyle
&\max_{\bm\alpha}~ &\bm\alpha^\top \bm G\bm\alpha,\\
&\text{s.t.} & \alpha_0 = 1,\\
&&\alpha_i\in \{0,1,\ldots, L-1\},\qquad \forall i=1,\ldots,d,\\
&&\bm{A}(\alpha_1,\ldots,\alpha_d)^\top\leq \bm{b}.
\end{array}
\end{equation}
Note that this is  a  maximization program with convex quadratic objective function.
For the ``unconstrained" case (i.e., without side constraints $\bm{A}(\alpha_1,\ldots,\alpha_d)^\top\leq \bm{b}$), we may relax the integrality constraints, and an optimal solution will be at a vertex of the box constraints. In this case, there are only $2^d$ possible solutions to check, rather than $L^d$ possible integer solutions.
This means that we will only generate experiments using the extreme levels, $0$ and $L-1$. On the other hand, for the first-order model, we have $L=2$, and all levels are extreme.
Therefore, in this unconstrained case, the problem \eqref{quad_sub} is equivalent to a ``boolean quadratic program.''  For modest values of $d$ (e.g., $d = 12$), it is feasible to solve \eqref{quad_sub} by enumeration. However, for larger $d$ (e.g., $d = 20$), enumeration becomes impractical as the solution space grows too large to explore in a reasonable time. To address this, we formulate \eqref{quad_sub} as a linear integer program by linearizing the quadratic terms, allowing us to efficiently solve it using \texttt{Gurobi}. In addition, we provide primal heuristics that speed up the search for the optimal solution for \eqref{quad_sub}.


For the second-order model, we arrive at a degree four integer optimization problem over integer points in a d-dimensional polytope. We address this by linearizing all higher degree terms in the objective and solving the resulting higher-dimensional linear integer program using \texttt{Gurobi}. 

Finally, we also consider solving \eqref{pricing_LS} heuristically before solving it exactly using a solver. We employ a local search method with an initial solution $\bar\alpha=(1; \bfx')$, where the criterion for the improvement of a given solution is the increase in the value of the objective function of \eqref{pricing_LS} and different neighborhoods are explored, defined by flipping over one bit of the current solution, termed ``Bit Flip," or swapping two bits, termed ``Bit Swap." More specifically, for bit flip, the neighborhoods of a given solution $\bar{\bm{\alpha}}$ with $\bar{\alpha}_0=1$ are  defined  as
\begin{align} \label{eq:heuristic-1} \tag{Bit Flip}
&\mathcal{N}_1(\bar{\bm\alpha}):= \left\{\bm\alpha\in\mathbb{Z}^{d+1}: 
\alpha_0 = 1,~
0\leq \alpha_i\leq L-1, ~  \mbox{ for all }i\in[d], ~ 
A(\alpha_1,\ldots,\alpha_d)^\top\leq b,\right. \nonumber \\[-5pt]
& \qquad
\left.(\alpha_1,\ldots,\alpha_d)^\top-(\bar{\alpha}_1,\ldots,\bar{\alpha}_d)^\top=\pm\bm{e}_i, \mbox{ for some } i\in[d]
\vphantom{\mathbb{Z}^{d+1}}
\right\}, \nonumber
\end{align}
and for bit swap, they are defined as
\begin{align} \label{eq:heuristic-2} \tag{Bit Swap}
&\mathcal{N}_2(\bar{\bm\alpha}):= \left\{\bm\alpha\in\mathbb{Z}^{d+1}: 
\alpha_0 = 1,~
0\leq \alpha_i\leq L-1, ~  \mbox{ for all } i\in[d], ~ 
A(\alpha_1,\ldots,\alpha_d)^\top\leq b,\right. \nonumber \\[-5pt]
& \qquad
\left.(\alpha_1,\ldots,\alpha_d)^\top-(\bar{\alpha}_1,\ldots,\bar{\alpha}_d)^\top=\bm{e}_i - \bm{e}_j, \mbox{ for some } i,j\in[d], i\neq j
\vphantom{\mathbb{Z}^{d+1}}
\right\}, \nonumber
\end{align}
where $\bm{e}_i$ is the $i$-th standard basis.

In the next section, we provide a theoretical analysis of a polynomial-time implementation for an approximate local search algorithm.

\section{Theoretical results for D-optimal design.} \label{sec:local-theory}
In this section, we prove Theorem~\ref{theorem:approx-local} and Corollary~\ref{corollary:full-linear-guarantee}. The first result gives a guarantee of the output of the local search algorithm if the polynomial optimization problem  \eqref{pricing_LS} described in Section \ref{sec:Pricing} can be solved approximately. We also show that this result can be applied to obtain an approximation guarantee when set  $\mathcal{Y} = \{-1, 1\}^d$. In addition, we prove the hardness for the sub-problem $\max_{ \bfv \in \A} \bfv^\top \bm{G} \bfv$ arising from a second-order problem.  

\subsection{Local search with approximate local optima.}
In this section, we analyze the local search algorithm when it finds an approximate local solution that can be efficiently implemented in polynomial time, as stated in the following Theorem~\ref{theorem:approx-local}. More specifically, in Theorem \ref{theorem:approx-local}, we prove that instead of identifying the best improving move by solving an integer program (IP), an algorithm that yields an approximately improving solution still achieves an approximately locally optimal solution. Besides, we show that for the special case when the set of allowed vectors is $\mathcal{Y}=\{\bfx\in \{0,1\}^d: x_1=1\}$, Nesterov's result~\cite{nesterov2011random} can be used to give an efficient polynomial time algorithm for implementing the improvement step approximately. We also provide a performance guarantee on the approximate locally optimal solution.

We first prove Theorem~\ref{theorem:approx-local}. The proof follows the outline of Madan et al. \cite{madan2019combinatorial} where they analyze the local search algorithm. In Theorem~\ref{theorem:approx-local}, we assume that the improvement step in the local search algorithm can be approximated within a factor of $\frac{1}{\rho}$. 

A crucial ingredient in the analysis is the convex relaxation of \eqref{CD-Opt2} and its Lagrangian dual as defined below. 
First, the continuous relaxation of the D-opt design problem \eqref{CD-Opt2} is given by
\begin{align}
\phi_R=\max_{\bm{\lambda}}\left\{\ln\det\left(\sum_{i = 1}^n\lambda_i \bfv^i{\bfv^i}^\top\right): \sum_{i = 1}^n\lambda_i=k,\lambda_i\geq 0,\forall i\in [n] \right\}. \label{eq_d_opt_combinatorial_discrete_exp}
\end{align}
Next, let us consider its equivalent form
\begin{align*}
\phi_R=\max_{\bm{\lambda},\bm Y}\left\{\ln\det\left(\bm Y\right): \sum_{i = 1}^n\lambda_i\bfv^i{\bfv^i}^\top \succeq \bm Y, \sum_{i = 1}^n\lambda_i=k,\lambda_i\geq 0,\forall i\in [n] \right\}, 
\end{align*}
which admits the following Lagrangian dual with dual multiplier $\bm\Lambda$ associated with constraint $\sum_{i = 1}^n\lambda_i\bfv^i{\bfv^i}^\top\succeq \bm Y$ and dual multiplier $\nu$ associated with constraint $\sum_{i = 1}^n\lambda_i=k$:
\begin{align}
\phi_R^{D}=\min_{\bm{\Lambda}\succeq \bm 0,\nu}\left\{k\nu-\ln\det\left(\bm \Lambda\right)-p: {\bm{v}^i}^\top \bm\Lambda\bfv^i\leq\nu,\forall i\in [n] \right\}, \label{lagrangian_dual_relaxation}
\end{align}
whose optimal value is equal to that of the continuous relaxation \eqref{eq_d_opt_combinatorial_discrete_exp} due to the strong duality of the convex program with the Slater condition \cite{ben2001lectures}.

\thmlocal*
\begin{proof} To prove the theorem, we show a stronger inequality by showing the given guarantee using $\phi_R$ instead of $\phi$, which immediately implies the theorem since $\phi_R \geq \phi$. First, we show that the $(\frac{1}{\rho})$-approximation algorithm for solving the problem $\max_{ \bfv \in \A} \bfv^\top \bm{G} \bfv$ can be used to implement the local search algorithm where each local step can be solved approximately within a factor of $1/\rho$ as shown in the following lemma. 
\begin{lemma} \label{claim:local-opt}
Given some $\rho>1$, if there exists a $\frac{1}{\rho}$-approximation algorithm to the problem $\max_{\bfv \in \A}\bfv^\top \bm{G}\bfv$ for all $\bm{G} \succeq 0$, then there is a $\frac{1}{\rho}$-approximation algorithm to the problem $\max_{\bfv^i \in \hat{S}, \bfv^j \in \mathcal{Y} }\det(\bm{S}_{-i} + \bfv^j{\bfv^j}^\top)$.
\end{lemma}
\begin{proof}
 For each $i \in \hat{S}$, if $\bm{S}_{-i}$ has an inverse, then $\det( \bm{S}_{-i} + \bfv^j{\bfv^j}^\top) = \det( \bm{S}_{-i})(1 + {\bfv^j}^\top\bm{S}_{-i}^{-1}\bfv^j)$. Thus, we can use the approximation algorithm to solve $\max_{ \bfv \in \A  }\bfv^\top\bm{S}_{-i}^{-1}\bfv$ to get a vector $\bfv'$ satisfying $\det(\bm{S}_{-i} + \bfv'{\bfv'}^\top) = \det(\bm{S}_{-i})(1 + \bfv'\bm{S}_{-i}^{-1}\bfv') \geq  \det(\bm{S}_{-i})(1 +  \frac{{\bfv^*}^\top\bm{S}_{-i}^{-1}\bfv^*}{\rho}) \geq \frac{\det(\bm{S}_{-i})}{\rho}(1 +  {\bfv^*}^\top\bm{S}_{-i}^{-1}\bfv^*) = \frac{1}{\rho}\max_{j}\det(\bm{S}_{-i} + \bfv^j{\bfv^j}^\top)$. Similarly, if $\bm{S}_{-i}$ does not have an inverse, then we have $\det( \bm{S}_{-i} + \bfv^j{\bfv^j}^\top ) = \det( \bm{S}_{-i}) {\bfv^j}^\top(\bm{I}_p - \bm{S}_{-i}^{\dagger}\bm{S}_{-i} )\bfv^j$ and can apply the approximation algorithm to $\bm{I}_p - \bm{S}_{-i}^{\dagger}\bm{S}_{-i}$. Since we can get a $\frac{1}{\rho}$-approximation for each fixed $i$ we can also get the same approximation when we maximize over $i$. 
 \QEDB
\end{proof}

We show the bound in the theorem by constructing a feasible dual solution $(\bar{\bm{\Lambda}}, \nu)$ to the dual program~\eqref{lagrangian_dual_relaxation}, whose objective value is comparable to the objective value of the output of the approximate local search algorithm.

Let $(\bar{\bm{\Lambda}}, \nu)$ where $\bar{\bm{\Lambda}} = (\alpha \bm{S})^{-1}$ for $\alpha = \left( \frac{p}{k} \frac{k - p + 1}{p + k(\rho - 1)}  \right)$ and $\nu = \frac{p}{k}$ be a candidate dual solution.  This allows us to relate the objective value of $\bm{S}$ to $\phi_{R}$ since 
\[ 
\phi_{R} \leq k\nu - \ln( \det( (\alpha \bm{S})^{-1} )  ) - p. 
\] 

The inequality follows since the right-hand side is the objective value of the dual solution $(\bar{\bm{\Lambda}}, \nu)$ and the dual is a minimization problem. 

We are ready to show the approximation guarantee. Note that 
 \begin{align*}
 \frac{\phi_R}{p} &\leq \ln(\alpha) + \frac{\ln(\det(\bm{S}))}{p} + \frac{k\nu}{p} - 1 
 = \ln\left( \frac{p}{k} \frac{k-p + 1}{p + k(\rho - 1)}\right) + \frac{\ln(\det(\bm{S}))}{p}.
 \end{align*}
 The first equality follows since the objective value of the dual solution $(\bar{\bm{\Lambda}}, \nu)$ is given by $k\nu - \ln( \det(\bar{\bm{\Lambda}})) - p$ and the second equality is from the definition of $\nu$ and $\alpha$.

Thus, it remains to show that the candidate dual $(\bar{\bm{\Lambda}}, \nu)$ is feasible to the dual problem \eqref{lagrangian_dual_relaxation}.
First, observe that $\bar{\bm{\Lambda}} \succeq 0$ since $\bm{S}\succeq 0$ and $\alpha > 0$. We now verify the constraint 
${\bm{v}^i}^\top \bar{\bm{\Lambda}}\bfv^i\leq\nu$ for each $i$ for which we need the following notation.
\begin{enumerate}[(I)]
 \item For $1 \leq i \leq n$, $\tau_i := {\bfv^i}^\top\bm{S}^{-1}\bfv^i$
 \item For $1 \leq i, j \leq n$, $\tau_{ij} :={\bfv^j}^\top\bm{S}^{-1}\bfv^j$.
\end{enumerate}
The next two lemmas are useful to prove the feasibility of the dual solution. 
\begin{lemma} \label{claim:tau}
 For any $i \in I$ and $1 \leq j \leq n$ we have $\tau_j - \tau_i\tau_j + \tau_{ij}\tau_{ji} \leq \tau_i + \rho - 1$.
\end{lemma}
\begin{proof}
 Since $\bm{S}$ is an approximate locally optimal solution, we have that $\det(\bm{S} - \bfv^i{\bfv^i}^\top + \bfv^j{\bfv^j}^\top ) \leq \rho \det(\bm{S})$. Following exactly the calculations of lemma 9 in \cite{madan2019combinatorial}, we get that 
 $\rho \det(\bm{S}) \geq \det(\bm{S})(1 + {\bfv^j}^\top\bm{S}^{-1}\bfv^j)(1 - {\bfv^i}^\top(\bm{S} +  \bfv^j{\bfv^j}^\top)^{-1}\bfv^i)$ which is equivalent to $(1 + \tau_j)(1 - \tau_i + \frac{\tau_{ij}\tau_{ji}}{1 + \tau_j}) \leq \rho$. This concludes the proof since the last inequality is equivalent to the inequality of the claim.  
 \QEDB
\end{proof}

\begin{lemma} \label{claim:tau2-f}
 For any $1 \leq j \leq n$, we have $\tau_j \leq \frac{p + k(\rho - 1)}{k-p + 1}$.
\end{lemma}
\begin{proof}
 Summing the inequality for all $i \in I$ from Lemma \ref{claim:tau}, we have
 \begin{align*}
  \sum_{i \in \hat{S}} (\tau_j - \tau_i \tau_j + \tau_{ij}\tau_{ji}) &\leq \sum_{i \in \hat{S}}\tau_i + k(\rho-1)  = p + k(\rho -1).
 \end{align*}
 We also have the following expression for the left-hand side: $\sum_{i \in I} (\tau_j - \tau_i \tau_j + \tau_{ij}\tau_{ji}) = k\tau_j - p \tau_j + \tau_j$. Combining this with the upper bound concludes the proof.
 \QEDB
\end{proof}
To show ${\bm{v}^i}^\top \bar{\bm{\Lambda}}\bfv^i\leq\nu$ for all $i \in [n]$ , it is equivalent to show $\max_{i = 1}^n {\bm{v}^i}^\top \bar{\bm{\Lambda}}\bfv^i \leq \nu = \frac{p}{k}$. We have, 
\begin{align*}
    \max_{i = j}^n {\bm{v}^j}^\top \bar{\bm{\Lambda}}\bfv^i &= \frac{1}{\alpha} \max_{j = 1}^n\tau_j \\
    & \leq \frac{p}{k}.
\end{align*}
The first equality follows by the definition of $\tau_j$, and the inequality follows by applying the upper bound in Lemma~\ref{claim:tau2-f}. 
\QEDA
\end{proof}

We note that as in Madan et al. ~\cite{madan2019combinatorial}, the guarantee given in the above lemma is more meaningful when $k \gg p$. Now we apply Theorem \ref{theorem:approx-local} to prove Corollary~\ref{corollary:full-linear-guarantee}. 

\corrNest*
\begin{proof}
We prove the corollary by using Theorem~\ref{theorem:approx-local} when $\mathcal{Y}' = \{ -1, 1\}^d$ and then transforming the $\{-1, 1\}$ solution to a $\{0, 1\}$ solution. Nesterov's algorithm gives a $\frac{2}{\pi}$-approximation to the problem $\max_{\bfx \in \mathcal{Y}'} \bfx^\top \bm{G} \bfx$ for any $\bm{G} \succeq 0$ with $\bm{G} \in \mathbb{R}^{d \times d}$. This satisfies the requirements of Theorem~\ref{theorem:approx-local} and allows us to apply it with $\rho = \frac{2}{\pi}$. Let $\phi'$  be the optimal ($\ln \det$) value of the integral solution. Then, from Theorem~\ref{theorem:approx-local}, we get a solution $\bm{S}'$ where $\bm{S} = \bm{V}'^{\top}\bm{V}'$ for $\bm{V}' \in \{-1, 1\}^{p \times k}$ satisfying 
\begin{align}
    \det(\bm{S}') \geq e^{\phi'}  \left(  \frac{k - p + 1}{k} \frac{p}{p + k (\frac{\pi}{2} - 1)} \right)^p \label{eq:appx_1}
\end{align}  

We will show that the $\{-1 ,1\}$ solution $\bm{S}'$ can be transformed to a $\{0, 1\}$ solution $\bm{S}$ such that $\det(\bm{S}) = \frac{\det(\bm{S}')}{2^{2(p - 1)}}$. This transformation and the given guarantee immediately proves the first inequality of the corollary, since it also implies that $\phi = \frac{\phi'}{2^{2(p-1)}}$.  Finally, we show that $\phi_R = \frac{\phi_R'}{2^{2(p - 1)}}$ where $\phi_{R}'$ is the optimal value of the convex relaxation for the set $\mathcal{Y}'$ which allows us to get a stronger version of the first inequality of the corollary where $\phi$ is replaced with $\phi_R$.  The transformation follows via a series of matrix operations that are standard for the Hadamard problem.

We will transform the matrix $\bm{V}'$ to a matrix $\bm{V} \in \{0, 1\}^{p \times k}$ such that the first row is all-one. We can do this with the following steps: 
 \begin{enumerate}[{Step} 1.]
     \item Multiply the columns by  of $\bm{V}'$ by $-1, 1$ so that the first row of $\bm{V}'$ is all ones.
     \item Add the first row of $\bm{V}'$ to all other rows.
     \item Divide all rows except the first row by $2$.
 \end{enumerate}
 We note that the first step can be achieved by multiplying by a ($k \times k$) diagonal matrix $\bm{D}$ on the right, and that steps 2 and 3 can be achieved by multiplying $\bm{V}'$ on the left by $p-1$ matrices $\bm{L}_1, \ldots, \bm{L}_{p-1}$. Thus, we have $\bm{V}  = \bm{L}_{p-1}\ldots \bm{L}_1 \bm{V}' \bm{D}$ so that $\det(\bm{L}_i) = 1/2$ for all $i = 1, \ldots, p-1$ . Thus, we can get a solution with vectors in $\mathcal{Y}=\{0,1\}^d$ by setting $\bm{S} = \bm{V} \bm{V}^\top$ with objective value give by 
 \begin{align*}
     \det(\bm{S}) &= \det \left( \bm{L}_{p-1}\bm{L}_2 \ldots \bm{L}_1 \bm{V}' \bm{D}^2  \bm{V}'^\top \bm{L}_1^\top \ldots \bm{L}_{p-1}^\top   \right) \\
     &= \det \left( \bm{L}_{d-1}\bm{L}_2 \ldots \bm{L}_1 \bm{V}'\bm{V}'^\top \bm{L}_1^\top \ldots \bm{L}_{p-1}^\top   \right) \\
     &= \frac{\det( \bm{S}') }{2^{2(p-1)}}. 
 \end{align*}

Then the proof of the corollary is concluded by the following lemma. 
\begin{lemma} \label{prop:relax-value-2}
    When $\mathcal{Y} = \{\bfx \in  \{0, 1\}^d : x_1 = 1 \}$ , an optimal solution to the continuous relaxation \eqref{eq_d_opt_combinatorial_discrete_exp} is $\lambda_i^*=k/2^{p-1}$ for each $i\in [2^{p-1}]$ with optimal value $e^{\phi_{R}}  = \frac{k^p}{2^{2(p-1)}}$.
 \end{lemma}
 \begin{proof}
 The proof is to construct primal and dual solutions that have the same objective value. We begin by constructing the primal solution. We define the primal solution by letting $\lambda_i = \frac{k}{2^{p-1}}$ for $i = 1, \ldots 2^{p-1}$. Let $\bm{A}$ be the matrix with all $2^{p-1}$ vectors from $\A$ as its columns and $\bm{D}$ be a diagonal matrix where $D_{ii} =  \lambda_i$ for $i \in [2^{p-1}]$. Then the objective value of the primal solution is given by $\ln (\det (  \bm{A} \bm{D} \bm{A}^\top ) )$. We now calculate $\det (  \bm{A} \bm{D} \bm{A}^\top )$. We observe that $\bm{A} \bm{D} \bm{A}^\top = k \begin{bmatrix}
1 & 
\bm{1}^\top /2\\
\bm{1}/2  & (\bm{J}_{p - 1} + \bm{I}_{p-1})/4
\end{bmatrix}$   where $\bm{1}$ is the all ones vector in dimension $p-1$ and $\bm{J}_{p - 1}$ is the all ones matrix in dimension $p-1$. Then we have that $\det(\bm{A} \bm{D} \bm{A}^\top) = k^p \det((\bm{J}_{p - 1} + \bm{I}_{p-1})/4 - \bm{1} (\bm{1}^\top)/4 ) = k^p \det( \bm{I}_{p-1} /4) = \frac{k^p}{2^{2(p-1)}} $. Now we construct a feasible dual solution with the same objective value. We set $\bm{\Lambda} = (\bm{A} \bm{D} \bm{A}^\top)^{-1} $ and $\nu = p/k$. Then the objective value of the dual solution, $k\nu - \ln(\det(\bm{\Lambda})) - p$,  is equal to the primal objective by our calculation of $\det( \bm{A} \bm{D} \bm{A}^\top )$. We now verify the feasibility of the dual solution. The matrix $\bm{\Lambda}$ is PSD since its inverse is PSD, so it remains to show that  $\bfv^\top \bm{\Lambda} \bfv \leq \nu = p/k$ for all $\bfv \in \A$. We can use the block structure of $\bm{A} \bm{D} \bm{A}^\top$ to get that $\bm{\Lambda} =  \frac{1}{k} \begin{bmatrix}
p & 
-2\bm{1}^\top\\
-2\bm{1} &  4\bm{I}_{p-1}
\end{bmatrix}$. Then it is easy to verify that   $\bfv^\top \bm{\Lambda}\bfv = p/k = \nu$ for all $\bfv \in \A$.
     \QEDB
 \end{proof}
\QEDA
\end{proof}

\subsection{Hardness of polynomial optimization problem.}
So far, we have discussed two approaches to solve the problem $\max_{\bfv \in \A} \bfv^\top \bm{G}\bfv$ for $\bm{G} \succeq 0$: in Section \ref{sec:Pricing}, we describe how to solve the problem using an integer programming approach exactly whose running time can be exponential in the input size. In Corollary~\ref{corollary:full-linear-guarantee}, we show that for the special case of a full first-order model when $\A=\mathcal{Y}=\{\bfx\in \{0,1\}^d: x_1=1\}$, we have an efficient $\frac{2}{\pi}$-approximation algorithm. This algorithm is based on the $\frac{2}{\pi}$-approximation by Nesterov~\cite{nesterov2011random}  to solve the problem $\max_{\bfv \in \A} \bfv^\top \bm{G}\bfv$ for $\bm{G} \succeq 0$ where $\A=\mathcal{Y}=\{-1,1\}^d$. A natural question is whether such approximation algorithms also hold for general $\A$, in particular for second-order models. 

In this section, we show a negative result and prove that it is NP-hard to get an approximation better than $1/2$ when set $\A$ is defined by a special case of the second-order model. In this case, hardness even follows when the matrix $\bm{G}$ is diagonally dominant and $\mathcal{M}$ has all degree-one monomials and a limited number of degree-two monomials. 
The diagonally dominant matrix is of special interest, since Goemans and Williamson \cite{GW95} show that their algorithm for max-cut also achieves a $0.878$-approximation for the problem $\max_{\bfx \in \{-1, 1\}^d}\bfx^\top \bm{G}\bfx$ that was the basis for the algorithm by Nesterov~\cite{nesterov2011random} for the general positive semi-definite matrix $\bm G$. For a diagonally dominant matrix $\bm{G}$ our goal is to solve the problem 
\begin{align} \label{eq:diag-dominant}
 \max_{\bfv \in \A} \bfv^\top \bm{G} \bfv=\max_{\bfx \in \{-1, 1\}^d} \bfp(\bfx)^\top \bm{G} \bfp(\bfx).
\end{align}
where $\bfp(\bfx)$ has coordinates corresponding to all monomials of degree 1 and some monomials of degree 2.
 
We use the problem $\mathrm{MAX-3LIN}$ for a reduction. In $\mathrm{MAX-3LIN}$ we are given $m$ equations over variables $z_1, \ldots, z_p$ of the form $z_i + z_j + z_k \equiv 0 \mod 2$ or   $z_i + z_j + z_k \equiv 1 \mod 2$ and we need to select values of $z_i \in \{0, 1\}$ for each $i\in [p]$ to maximize the number of equations that are satisfied. The instance \emph{cannot} contains two distinct equations involving the same set of variables. It has been shown that $\mathrm{MAX-3LIN}$ is  NP-hard to approximate better than $1/2$ \cite{haastad2001some}.
\begin{proposition}
 There exists an approximation preserving polynomial-time reduction from $\mathrm{MAX-3LIN}$ to Problem \eqref{eq:diag-dominant}, and thus, it is NP-hard to approximate it better than a factor of $\frac12$.
\end{proposition}
\begin{proof}
  We observe that $z_i + z_j + z_k \equiv 0 \mod 2$ is satisfied if and only if $\{z_i, z_j, z_k\}$ contains an even number of ones. Similarly, $z_i + z_j + z_k \equiv 1 \mod 2$ if and only if  $\{z_i, z_j, z_k\}$ contains an odd number of ones. 
  Given the set of $m$ equations, we now construct an instance of problem \eqref{eq:diag-dominant} by defining the set of monomials $\mathcal{M}$ and the matrix $\bm{G}$. Our reduction will have the following property:
 \begin{itemize}
  \item For all $\bm{x} \in \{-1, 1\}^d$ the value $p(\bm{x})^{\top} \bm{G} p(\bm{x})$ is equal to the number of equations satisfied by the $\mathrm{MAX-3LIN}$ solution $z_i = \frac{x_i + 1}{2}$ for all $i \in [p]$.
 \end{itemize}

 Firstly, we let $\mathcal{M}$ contain all degree-one monomials meaning $m_i(\bfx) = x_i$ for $i = 1, \ldots, d$. For every equation $z_i + z_j + z_k$ with $i < j < k$ we add a degree-two monomial corresponding to pair $\{j,k\}$ meaning $m_{jk}(\bfx) = x_jx_k$. Let $\mathcal{C}$ be a set of subsets that defines the set of all monomials meaning for all $F \in \mathcal{C}$, $m_{F}(\bfx)= \prod_{i \in F}x_i \in \mathcal{M}$. We will use the elements $\mathcal{C}$ to index entries of the matrix $\bm{G}$, which will have dimensions $|\mathcal{C}| \times |\mathcal{C}|$. For every equation $z_i + z_j + z_k \equiv 0 \mod 2$ with $i < j < k$ we set the entry $G_{i, \{j,k \}} = G_{\{j,k \}, i} = -1/4$ and for $z_i + z_j + z_k \equiv 1 \mod 2$ we set the entry $G_{i, \{j,k \}} = G_{\{j,k \}, i} = 1/4$. All other off-diagonal entries are set to $0$. We define the diagonal entries by setting $G_{F_1, F_1} = \sum_{F_2 \in \mathcal{C}, F_2 \neq F_1}|G_{F_1, F_2}|$ for all $F_1 \in \mathcal{C}$. Thus, by construction, $\bm{G}$ is symmetric, diagonally dominant, and therefore also positive semi-definite. We define sets $\mathcal{C}_1,\mathcal{C}_2$ where $\mathcal{C}_1$ contains all $(i, j, k)$ such that $z_i + z_j + z_k  \equiv 0 \mod 2$ is an equation in the $\mathrm{MAX-3LIN}$ instance and $\mathcal{C}_2$ contains all $(i, j, k)$ such that $z_i + z_j + z_k  \equiv 1 \mod 2$ is an equation. For any $\bfx \in \{ -1, 1\}^d$, we have
 \begin{align*}
  \bfp(\bfx)^\top \bm{G} \bfp(\bfx) =  \frac{1}{2}\sum_{(i, j,k) \in \mathcal{C}_2  } \left(1 + x_ix_jx_k \right)+ \frac{1}{2}\sum_{(i, j,k) \in \mathcal{C}_1  } \left( 1 - x_ix_jx_k\right). 
 \end{align*}
 Then we observe that 
 \begin{align}
\frac{1 + x_jx_jx_k}{2} = \begin{cases} 1 & \text{if } x_i, x_j,x_k \text{ has an odd number of } -1 \\ 
 0 & \text{ otherwise} \end{cases} \label{eq:hardness-1}
 \end{align}
  and 
  \begin{align}\frac{1 - x_jx_jx_k}{2} = \begin{cases} 1 & \text{if } x_i, x_j,x_k \text{ has an even number of } -1 \\ 
 0 & \text{ otherwise} \label{eq:hardness-2}
 \end{cases}.\end{align}
 
 Then the vector $\bfx=(x_1, \ldots, x_d)$ can be mapped to a $\mathrm{MAX-3LIN}$ solution $\bfz=(z_1, \ldots, z_d)$ by setting $z_i= \frac{x_i + 1}{2}$. The mapping  from $\bfx$ and $\bfz$ has the property that $p(\bm{x})^\top \bm{G} p(\bm{x})$ is equal to the number of equations $\bfx$ satisfies in the $\mathrm{MAX-3LIN}$ instance. This follows by equations \eqref{eq:hardness-1}, \eqref{eq:hardness-2} and the fact that $x_i = -1$ implies $z_i = 0$ and $x_i = 1$ implies $z_i = 1$.
\QEDA
\end{proof}

\section{Continuous relaxation.} \label{sec:relaxation}
In this section, we consider the continuous relaxation \eqref{eq_d_opt_combinatorial_discrete_exp} discussed in Section \ref{sec:local-theory} and give an algorithm to solve it to optimality using column generation. We use both the primal~\eqref{eq_d_opt_combinatorial_discrete_exp} and dual~\eqref{lagrangian_dual_relaxation} to solve the problem efficiently.

The algorithm works by maintaining a set of candidate vectors $\A'\subseteq \A$. We initialize the set $\A'$ as a random subset of $\A$ of small cardinality. We solve the restricted primal obtained by setting all variables for vectors not in $\A'$ to zero and the corresponding restricted dual to optimality. We give the restricted primal and dual below, and for a set $\A'$, we refer to the restricted program as $R(\A')$.   
\begin{align} \displaystyle
\max_{\bm{\lambda}} &\quad \ln ( \det ( \sum_{\bfv \in \A'} \lambda_{\bfv} \bfv \bfv^\top ))    & \min_{\bm{\Lambda},\nu}  & \quad k\nu - \ln( \det(\bm\Lambda)) - p  & \label{eq:restricted-program} \\ 
\text{s.t.} &\quad \sum_{\bfv \in  \A'} \lambda_{\bfv} = k  &\text{s.t.} &\quad \bfv^\top \bm\Lambda \bfv \leq \nu , \quad \forall \bfv \in \A' \nonumber  \\ 
&\quad\lambda_{\bfv}  \geq 0, \quad\, \forall \bfv \in \A' &&\quad \bm\Lambda  \succeq 0. \:\:\ \nonumber \end{align}
In the following lemma, we show that it is straightforward to obtain an optimal dual solution from an optimal primal solution.  

\begin{lemma}\label{lem_dual_optiaml_restrict}
Let $\lambda_{\bfv}$ for all $\bfv \in \A'$ be an optimal solution for the restricted primal with set of vectors, $\A'$. Then the dual solution  $\bm\Lambda = (  \sum_{\bfv \in \A'} \lambda_{\bfv} \bfv \bfv^\top )^{-1}$ and $\nu = \max_{\bfv \in \A'} \bfv^\top \bm\Lambda \bfv$ is an optimal solution for the restricted dual with set of vectors $\A'$. 
\end{lemma}

\begin{proof}
    The proof follows by taking the Lagrangian multipliers for the primal program. We reformulate the primal program as follows, 
    \begin{align*}
 \max_{\bm{\lambda},\bm{X}} \quad &\ln ( \det (\bm{X} )) &      \\
\text{s.t.}\quad & \sum_{\bfv \in \A} \lambda_{\bfv} \bfv \bfv^\top \preceq \bm{X}   \\
 &\sum_{\bfv \in \A }\lambda_{\bfv} = k  \nonumber \\
 &\lambda_{\bfv} \geq 0,\forall \bfv \in \A' \nonumber.
\end{align*}
We use the multiplier $\Lambda$ for the first constraint and $\nu$ for the equality constraint. Then, the Lagrangian is given by 
\[ \max_{ \bm{X}, \bm{\lambda} \geq 0, \bm{\lambda}^\top \bm{1} = k} \min_{\nu, \bm\Lambda \succeq 0}  \ln ( \det( \bm{X} ) ) + \nu ( k-\sum_{\bfv \in \A }\lambda_{\bfv} )   + \sum_{\bfv \in \A } \lambda_{\bfv} \bfv^\top \bm\Lambda \bfv - \langle \bm\Lambda, \bm{X}\rangle.  \]
Taking the derivative with respect to $\bm{X}$ gives us $(\bm{X}^{-1})^\top  - \bm\Lambda = 0$ implying that $\bm\Lambda = \bm{X}^{-1}$. Hence, there exists an optimal dual solution that must also satisfy $\bm{X} =  \sum_{\bfv \in \A} \lambda_{\bfv} \bfv \bfv^\top$ and $\nu=\max_{\bfv \in \A'} \bfv^\top \bm\Lambda \bfv=p/k$. This concludes the proof of the lemma.\QEDA
\end{proof}

After solving the restricted primal and dual programs \eqref{eq:restricted-program} to optimality, in the next step, the algorithm checks whether the optimal dual solution to the restricted dual program is feasible for the original dual problem. If it is, we obtain the result that the current primal and dual solutions are also optimal for the original problem. Otherwise, we find a violating constraint to dual program corresponding to some vector $\bfv\in  \A\setminus \A'$. We then add $\bfv$ to $\A'$ and iterate in solving the updated formulated restricted primal and dual programs, which means that we move from the restricted program $R(\A')$ to $R(\A' \cup \{\bfv\})$. 
The main challenge is to solve the \emph{pricing problem} to find out whether there is a violating constraint. Given an optimal dual solution $(\bar{\Lambda}, \bar{\nu})$ to the restricted dual, the problem of checking the feasibility of this solution to the original program is equivalent to solving 
$$\max_{\bfv \in \A} \bfv^\top \bar{\bm\Lambda} \bfv $$ and checking if the maximum is at most $\bar{\nu}$ or not. This problem turns out to be exactly the problem \eqref{pricing_LS} that we have discussed in Section \ref{sec:Pricing}, and we apply the same methodology to solve it by reformulating it as a polynomial optimization problem. In the next proposition, we show that solving the pricing problem optimally also gives us a feasible dual solution. This is useful since it provides an upper bound for the value of the unrestricted program $R(\A)$ and therefore, it is also an upper bound of the value of the integral optimum solution. 

\begin{proposition} \label{proposition:ip-upper-bound}
    Let $(\bar{\bm\Lambda}, \bar{\nu})$ be a feasible solution to the restricted dual $R(\A')$ for a subset $\A' \subset \A$ and $\alpha = \max_{\bfv \in \A} \bfv^\top \bar{\bm\Lambda} \bfv $. Then $(\bar{\bm\Lambda}, \alpha)$ is a feasible dual solution to the unrestricted dual program  $R(\A)$. 
\end{proposition}
\begin{proof}
    This follows immediately, since $\bar{\bm\Lambda} \succeq 0$ and $\bfv^\top \bar{\bm\Lambda} \bfv \leq \alpha $ for all $\bfv \in \A$ by definition of $\alpha$. \QEDA
\end{proof}
We compute $\alpha$ in Proposition~\ref{proposition:ip-upper-bound} every time we solve the pricing problem optimally when running the column generation algorithm. In Section~\ref{sec:experiments}, where we present our experimental results, we utilize Proposition~\ref{proposition:ip-upper-bound} to derive an upper bound for large instances where obtaining the optimal value of $\R(\A)$ is computationally intractable.

To further speed up the implementation, we use additional steps that we describe in the following section.



\subsection{Further implementation steps.}\label{Section:colgen-implement}
In this section, we describe some extra steps that speed up the column generation algorithm. 
The full algorithm is described in Algorithm \ref{alg:col-gen}. 

\paragraph{Sparsification.} 
To ensure that the size of set $\mathcal{Y}'$ does not increase too much over iterations, we describe a sparsification procedure that reduces the size of set $\mathcal{Y}'$ without any loss to the objective value. This is useful since the formulation~\eqref{eq:restricted-program} depends on the size of $\A'$. In particular, whenever $|\A'|>\binom{p+1}{2}$, we can find a subset $\A'' \subset \A'$ of strictly smaller size such that the restricted problem $R(\A'')$ has the same objective as $R(\A')$. The search for $\A''$ can be reduced to solving a linear program described below. Given a feasible primal solution $\lambda^*_{\bfv}$ for all $\bfv \in \A'$ to the restricted problem \eqref{eq:restricted-program}, we formulate a linear programming feasibility problem with variables $\lambda_{\bfv}$ for all $\bfv \in \A'$. We ensure the matrix constraint that $\sum_{\bfv \in \A'} \lambda_{\bfv} \bfv {\bfv}^\top = \sum_{\bfv \in \A'} \lambda_{\bfv}^* \bfv {\bfv}^\top $ must hold, which is simply a set of linear constraints for variables $\lambda_{\bfv}$. Additionally, we impose that $\bm\lambda$ is feasible for $R(\A')$. Then, we obtain the following LP feasibility problem:
\begin{align}
 \min_{\bm{\lambda}} \quad &0   \label{LP:sparse}   \\
\text{s.t.}\quad &\sum_{\bfv \in \A' } \lambda_{\bfv} \bfv \bfv^\top = \sum_{\bfv \in \A' } \lambda_{\bfv}^*  \bfv \bfv^\top  \nonumber \\
 &\sum_{\bfv \in \A'  }\lambda_{\bfv} = k  \nonumber \\
 &\lambda_{\bfv} \geq 0,\forall \bfv \in \A' \nonumber.
\end{align}
We have the following theoretical bound on the size of the support of a continuous primal solution when we obtain a basic feasible solution. We then set $\A''=\{\bfv: \lambda_{\bfv}>0\}$ and continue.
\begin{proposition}
 For any set of vectors $\A'$ and weights $ \lambda^*_{\bfv} \geq 0$ for all $\bfv \in \A$ satisfying $\sum_{\bfv \in \A'  }\lambda^*_{\bfv} = k$ there exists $\lambda_{\bfv}$ for all $\bfv \in \A'$  satisfying $\sum_{\bfv \in \A'  }\lambda_{\bfv} = k $ such that the following two hold
 \begin{enumerate}
 \item $\sum_{\bfv \in \A' } \lambda_{\bfv} \bfv \bfv^\top = \sum_{\bfv \in \A' } \lambda^*_{\bfv} \bfv \bfv^\top$
 \item $| \{ \lambda_{\bfv} > 0 | \bfv \in \A' \} | \leq \binom{p}{2} + p + 1 $.
 \end{enumerate}
\end{proposition}
\begin{proof}
 The proof of this lemma follows by setting $\bm\lambda$ to any basic feasible solution of LP \eqref{LP:sparse} which exists since $\bm\lambda^*$ is a feasible solution to LP \eqref{LP:sparse}.  The first property holds due to the first constraint in the LP, and the second property holds since the LP has $\binom{p}{2} + p + 1$ constraints apart from non-negativity constraints. Thus, any basic feasible solution would have no more than $\binom{p}{2} + p + 1$ nonzero variables.
\QEDA
\end{proof}

\paragraph{Addition of random columns.} If the optimal value of \eqref{eq:restricted-program} is strictly less than the optimal value $\phi_R$ of the relaxation \eqref{lagrangian_dual_relaxation}, solving the pricing problem \eqref{pricing_LS} with an integer program guarantees that we obtain a violated column corresponding to a vector that is added to $\A'$. However, in practice, it can be slow to solve an integer program every iteration. A simple alternative approach is to add a small subset of random vectors from $\A$ to $\A'$ in addition to the vector that violates the dual constraint. This procedure can accelerate the solution procedure by saving many iterations via preemptively adding vectors without solving an integer program.

Thus, the high level idea of our algorithm can be summarized with the following steps. 
\begin{enumerate}[(i)]
 \item Solve the continuous relaxation of the restricted primal or dual problem \eqref{eq:restricted-program} to obtain a dual solution $(\bm\Lambda,\nu)$. Check if the current dual solution is feasible to the original problem \eqref{lagrangian_dual_relaxation} by solving the pricing problem \eqref{pricing_LS}. If not, then add a subset of random vectors from $\A$ to $\A'$ along with the vector that violates the dual constraint. 
 \item Sparsify the support of the current primal solution if the support is significantly larger than $p^2$. Return to Step 1.
\end{enumerate}

The details of the algorithm are described in Algorithm \ref{alg:col-gen}.
\begin{algorithm}[!t]
\caption{Pricing-based algorithm for the continuous relaxation of \ref{CD-Opt}}\label{alg:col-gen}
\begin{algorithmic}[1]
\Require $\A' \subset \A$  which is a set of vectors and $\bm{S} = \sum_{ \bfv \in \A' }\bfv {\bfv}^\top$ satisfying $\det(\bm{S}) > 0$. Tolerances $\delta, \epsilon, \gamma > 0$
\State Solve the primal program \eqref{eq:restricted-program} on vectors $\A'$ for  to get primal solution $\lambda$. 
\State Compute a feasible dual solution $\bm{\Lambda}, \nu$ from the primal solution.
\While{True}

\State Let $\bfv = \argmax_{\bfv \in \mathcal{Y}} {\bfv}^\top \bm{\Lambda} \bfv$ where the optimization problem is solved using a heuristic.
\If{${\bfv}^\top \bm{\Lambda} \bfv < (1 + \delta)\nu$}
 \State Let ${\bfv} = \argmax_{\bfv \in \mathcal{Y}} \bfv^\top \bm{\Lambda} \bfv$ where the optimization problem is solved using an IP. 
\EndIf

\If{ ${\bfv}^\top \bm{\Lambda} {\bfv} \leq (1 + \epsilon)\nu$ and  ${\bfv}$ was obtained from the IP}
 \Return $\A'$
\EndIf
\State $\A' \gets \A' \cup \{\bfv\}$

\If{$(\bm{\Lambda}, \nu)$ was obtained by solving the primal}

\State Add $p-1$ random vectors from $\mathcal{Y}$ to $\hat{S}$ 

\Else 
\State Add $2(p - 1)^2$ random vectors to $\hat{S}$ 

\EndIf
\If{$|\A'| > \lceil \frac{p^2}{3} \rceil$ and $(\bm{\Lambda}, \nu)$ was obtained by solving the primal}  
 \State Solve an LP to get $\A'' \subseteq \A'$ such that $ \sum_{\bfv \in \A''} \lambda_{\bfv} \bfv {\bfv}^\top  = \sum_{\bfv \in \A'} \lambda_{\bfv} \bfv {\bfv}^\top  $.
 \State $\A' \gets \A''$
\EndIf

\State Compute $(\bm{\Lambda}, \nu)$ using the new solution $\A'$.
\If{the improvement to the objective in the previous step is less than a factor of $\gamma$  }
\State Solve only the dual in future iterations.
\EndIf

\EndWhile
\end{algorithmic}
\end{algorithm}

\section{Experimental results.} \label{sec:experiments}

In this section, we show the efficacy and scalability of our approach for various classes of $D$-optimal design problems, in particular, first-order and second-order design, with knapsack constraints. We also compare our results with \texttt{JMP}~\cite{jmp2010design}, a popular commercial statistical software that can also compute the designs of the variants we study. In particular, we numerically compare the following variants of the $D$-optimal design problem.
\begin{enumerate}[(i)]
    \item  Full first-order model with two levels and a cardinality constraint. In this model, we have all vectors with no more than $r$ ones for some integer $r\leq d$, i.e., $\A=\mathcal{Y} = \{ \bfx \in \{0, 1\}^d: \quad   \bm{1}^\top \bfx   \leq r, x_1 = 1 \} $.
    \item    Full first-order model with two levels and two knapsack constraints. The knapsack constraints are give by two non-negative vectors $\bm{a}_1, \bm{a}_2 \in \mathbb{R}_{+}^{d}$ and two numbers $b_1, b_2 \geq 0$. Then the set of vectors is defined by $\A=\mathcal{Y} = \{ \bfx \in \{0, 1\}^d : \quad  x_1 = 1, \bm{a}_1 ^\top \bfx \leq b_1,  \bm{a}_2^\top  \bfx\leq b_2 \}$

\item We also consider the second-order model mixed with knapsack constraints. We recall the second-order model is defined by a set of monomials of degree at most two, given by set $\mathcal{M}$. That is, each design point is given by $\bfp(\bfx) = (m_1(\bfx), m_2(\bfx), \ldots m_p(\bfx) )^\top$, where $m_j(\bfx)$ is a monomial with degrees at most two. For vectors $\bm{a}_1, \bm{a}_2 \in \mathbb{R}_{+}^{d}$ and two numbers, $b_1, b_2 \geq 0$ get the following set of feasible vectors:
\[ \A = \{ (m_1(\bfx), m_2(\bfx), \ldots m_p(\bfx) )^\top\in \{0,1\}^{p} : \quad  \bfx \in \{0, 1\}^d, x_1 = 1, \bm{a}_1^\top \bfx   \leq b_1,  \bm{a}_2^\top \bfx  \leq b_2  \}.  \]
\end{enumerate}

\textbf{Experimental setup.} We ran the experiments on a Windows machine with an AMD EPYC Processor that has 16 cores. For each of the variants described at the beginning of this section, we ran both the local search algorithm described in Section~\ref{sec:localsearch} and the column generation algorithm described in Section~\ref{Section:colgen-implement}. We gathered various statistics on the runs of the two algorithms, such as total time, objective values, and time spent in \texttt{Gurobi}. Finally, we compared the results of our local search algorithm with \texttt{JMP}. We show a comparison between the objective values and also show that our local search algorithm can improve the \texttt{JMP} solution by using it as a starting point. The implementation is done with Python and is available online on GitHub \cite{githubDOPTDesign}.    

\subsection{First-order model with the cardinality constraints}

\begin{figure}[ht] 
\includegraphics[width=1.0\textwidth]{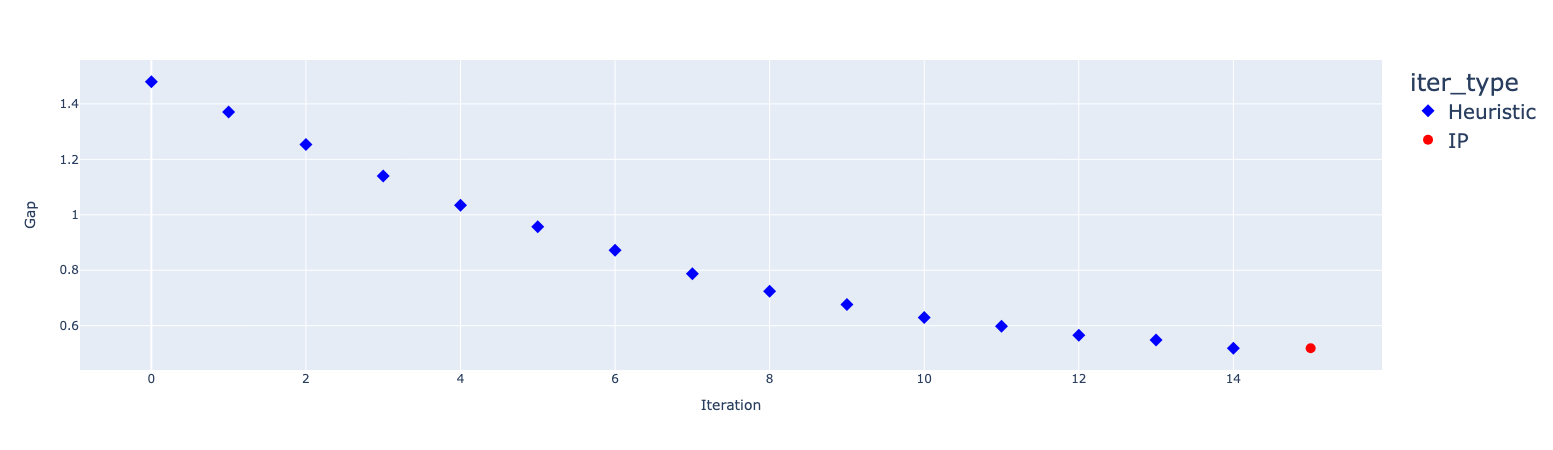}
\centering
\caption{A single run of the local search algorithm for solving the first-order model with the cardinality constraints for $d=18$. Heuristic means the local search found a locally improving move using the bit flip/swap heuristic, and IP means the algorithm failed to find an improving move with the heuristic and had to run \texttt{Gurobi}. In this example, the algorithm used the heuristic to find an improving move in every iteration, and the IP is only used in the last iteration to certify that the solution is a local optima.}\label{figure:local-cardinality}
\end{figure}

We begin with the simplest of the three variants-- the first-order model with the cardinality constraint. For the cardinality constraint, we use $r = \lfloor \frac{d}{3} \rfloor$. First, we begin with a single run of the algorithm for $ d=18$ in Figure~\ref{figure:local-cardinality}. The plot shows that for the simple cardinality constraint, the \eqref{eq:heuristic-1} and \eqref{eq:heuristic-2} heuristics work well to find a local improvement. The figure shows that the algorithm only needed to solve the IP in the last iteration to verify that the solution is a local optima. 

We also compare the results of the local search algorithm to \texttt{JMP} in Table~\ref{table:local_ones_values}. These objective values are for $\ln \det(.)$ and thus, a difference of $0.3$ in the objective value amounts to $e^{0.3}\simeq 1.35$ factor improving in the determinant objective.   Additionally, we are also able to improve the solution returned by \texttt{JMP} by using it as a starting point for our local search algorithm.
\begin{table}[!ht]
    \centering
    \caption{A comparison of different methods for solving the first-order model with the cardinality constraint. In this table, Relaxation Value is the value of the convex relaxation~\eqref{lagrangian_dual_relaxation}, and values are $\ln \det$. LS Value is the value of the solution the local search algorithm returned with a random starting solution, \texttt{JMP} Value is the value of the solution \texttt{JMP} returned, and \texttt{JMP} + LS gap is the value of the solution returned by the local search algorithm with the \texttt{JMP} solution as a starting point.}
    \begin{tabular}{|r|r|r|r|r|}
    \hline
        $d$ & LS Value & \texttt{JMP} Value & \texttt{JMP} + LS Value & Upper Bound (Convex Relaxation Value) \\ \hline
        11 & 13.641 & 13.299 & 13.641 & 14.189 \\ \hline
        12 & 18.968 & 18.938 & 18.964 & 19.270 \\ \hline
        13 & 20.785 & 20.521 & 20.800 & 21.085 \\ \hline
        14 & 22.641 & 22.129 & 22.645 & 22.897 \\ \hline
        15 & 27.451 & 27.122 & 27.434 & 27.781 \\ \hline
        16 & 29.437 & 29.095 & 29.415 & 29.895 \\ \hline
        17 & 31.374 & 30.910 & 31.422 & 32.003 \\ \hline
        18 & 36.325 & 35.878 & 36.360 & 36.844 \\ \hline
        19 & 38.639 & 38.014 & 38.694 & 39.189 \\ \hline
        20 & 41.115 & 40.221 & 41.056 & 41.528 \\ \hline
    \end{tabular}
    \label{table:local_ones_values}
\end{table}

Table~\ref{table:local_ones} shows the time it took to run the local search algorithm with a random starting solution. We can see from the table that the vast majority of the total time is spent on \texttt{Gurobi}. Specifically, the last iteration usually takes the most time, since it always requires $k$, the total number of points in the design,  \texttt{Gurobi} calls to verify that a solution is locally optimal.

\begin{table}[!ht]
    \centering
    \begin{tabular}{|r|r|r|r|r|}
    \hline
        $d$ & Total Time & \texttt{Gurobi} Time & \texttt{Gurobi} Iterations & Number of Iterations \\ \hline
        11 & 11.15 & 10.43 & 1 & 1 \\ \hline
        12 & 11.21 & 10.40 & 1 & 2 \\ \hline
        13 & 18.02 & 16.82 & 2 & 8 \\ \hline
        14 & 21.54 & 20.33 & 1 & 7 \\ \hline
        15 & 26.30& 23.97 & 1 & 12 \\ \hline
        16 & 27.54 & 25.79 & 1 & 13 \\ \hline
        17 & 29.73 & 28.12 & 1 & 9 \\ \hline
        18 & 30.57 & 28.15 & 1 & 16 \\ \hline
        19 & 30.65 & 28.71 & 1 & 17 \\ \hline
        20 & 82.53 & 78.89 & 2 & 32 \\ \hline
    \end{tabular}
\caption{Local search timing statistics for solving first-order model with the cardinality constraint. Each row of the table corresponds to a different single run of the algorithm for dimension $d$, and the columns give information on each particular run. Total Time is the total time the algorithm took from beginning to end in seconds, \texttt{Gurobi} Time is the total time the algorithm spent in running \texttt{Gurobi}, \texttt{Gurobi} Iterations is the total number of iterations where the algorithm called \texttt{Gurobi} to search for an improvement, and Iterations is the total number of iterations for that run.}
\label{table:local_ones}
\end{table}

\subsection{First-order model with knapsack constraints} \label{sec:experiment-knapsack}
In the next experiment, we study the first-order model with two knapsack constraints. To generate the two knapsack constraints, we first generated two vectors $\bm{a}_1, \bm{a}_2 \in \mathbb{R}^d$ where $a_{i1} = 0$ and $80\%$ of the entries are sampled uniformly at random from the set $\{0, \ldots, 5\}$ and the remaining entries are uniformly at random from $\{20, \ldots, 30\}$. The right-hand side the two constraints $b_i = \E[ \bm{a}_i^\top \bfx ]=1/2\sum_{j=2}^da_{ij}$ for $\bfx$ drawn uniformly at random from the set of vectors $\A=\mathcal{Y}= \{ \bfx \in \{0, 1\}^d :\quad x_1 = 1\}$ for $i=1,2$. 

\begin{figure}[h] 
\includegraphics[width=1\textwidth]{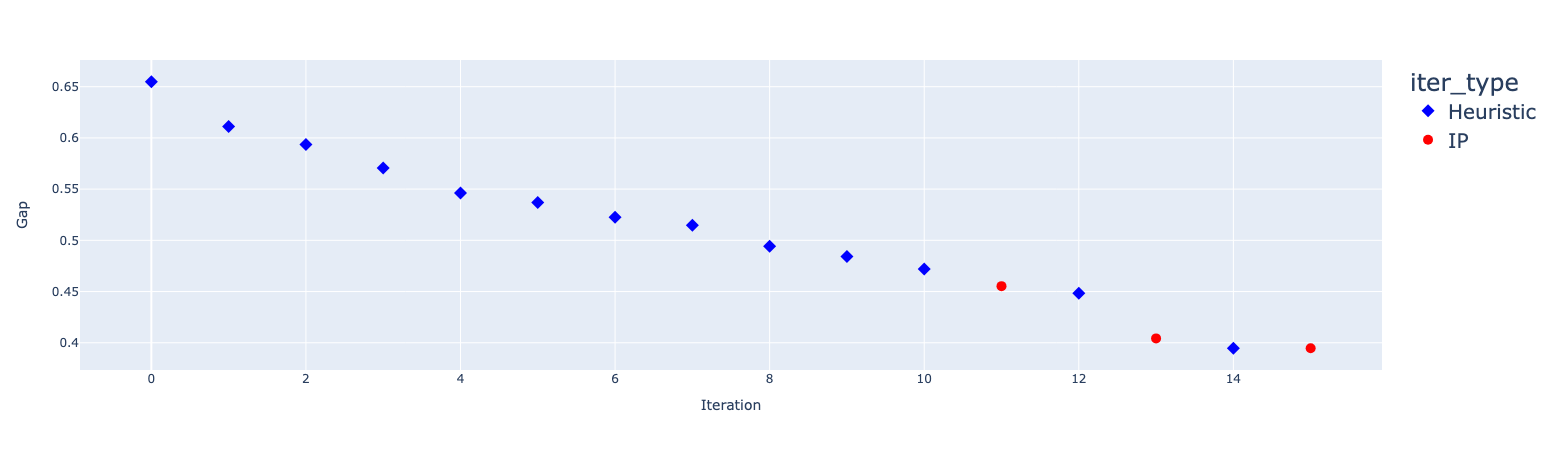}
\centering
\caption{A single run of the local search algorithm for solving first-order model with knapsack constraints for $d=13$. Heuristic means the local search found a locally improving move using the bit flip/swap heuristic, and IP means the algorithm failed to find an improving move with the heuristic and had to run \texttt{Gurobi} to solve the IP. }\label{figure:local-knap}
\end{figure}

Figure \ref{figure:local-knap} illustrates a single run of the local search. In this case, we solve the integer program in multiple iterations, in contrast to the first-order model with the cardinality constraint. In Figure~\ref{figure:local-knap}, we can see that the local search algorithm would have been stuck in iteration 10 using only (\ref{eq:heuristic-1}, \ref{eq:heuristic-2}) and that using the IP helps to return a better solution. 

Table~\ref{table:local_knap_values} compares the results of \texttt{JMP} and our algorithm. \texttt{JMP} usually gets a better objective compared to our algorithm with a random starting solution. Nonetheless, for most instances, we can further improve the \texttt{JMP} solution using our local search algorithm when we use the \texttt{JMP} solution as an initial solution. Figure~\ref{figure:local-knap-jmp} illustrates a single run in which we used the \texttt{JMP} solution as a starting point. Here, the first iteration was an IP iteration, since the \texttt{JMP} solution was locally optimal in the bit flip/swaps neighborhood.

\begin{table}[!ht]
    \centering
    \caption{A comparison of different methods for solving the first-order model with knapsack constraints.}
    \begin{tabular}{|r|r|r|r|r|}
    \hline
        d & LS Value & \texttt{JMP} Value & \texttt{JMP} + LS Value & Relaxation Value \\ \hline
        11 & 17.677 & 17.717 & 17.717 & 18.007 \\ \hline
        12 & 20.675 & 20.767 & 20.767 & 21.068 \\ \hline
        13 & 23.407 & 23.501 & 23.501 & 23.802 \\ \hline
        14 & 26.192 & 26.264 & 26.274 & 26.678 \\ \hline
        15 & 29.912 & 29.987 & 29.987 & 30.378 \\ \hline
        16 & 32.668 & 32.636 & 32.767 & 33.262 \\ \hline
        17 & 35.971 & 36.027 & 36.044 & 36.663 \\ \hline
        18 & 39.090 & 39.152 & 39.152 & 39.729 \\ \hline
        19 & 42.467 & 42.518 & 42.598 & 43.193 \\ \hline
        20 & 45.928 & 45.924 & 45.924 & 46.661 \\ \hline
    \end{tabular}
    \label{table:local_knap_values}
\end{table}

\begin{figure}[!ht] 
\includegraphics[width=17cm]{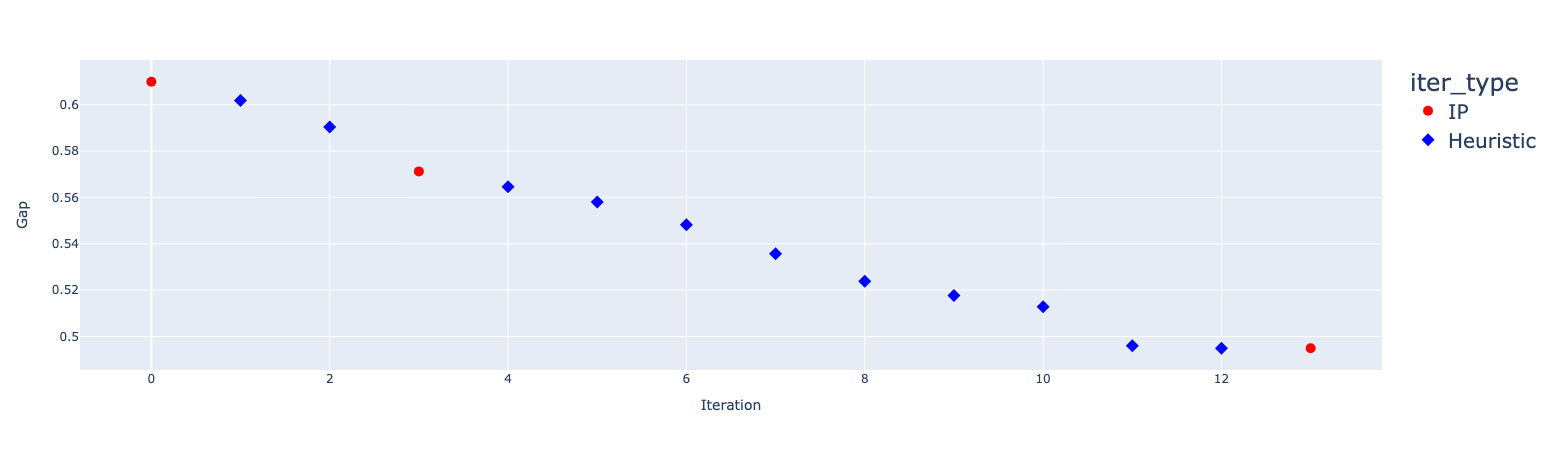}
\centering
\caption{A single run of the local search algorithm for solving first-order model with knapsack constraints for $d=16$. The local search is initialized with a solution given by \texttt{JMP}, and the IP is needed in the first iteration to find an improving move.  A single run of the local search algorithm. The Heuristic means the local search found a locally improving move using the bit flip/swap heuristic, and IP means the algorithm failed to find an improving move with the heuristic and had to run \texttt{Gurobi} to solve IP. The Gap is the difference between the value of the convex relaxation and the value of the solution. }
\label{figure:local-knap-jmp}
\end{figure}

Table~\ref{table:local_kn} displays the time it took to run the local search algorithm. Overall, the times are a bit higher in some runs than those with the cardinality constraint, since the \texttt{Gurobi} was executed in multiple iterations to solve IP. 

\begin{table}[!ht]
    \centering
    \caption{Local search timing statistics for solving first-order model with knapsack constraints.}
    \begin{tabular}{|r|r|r|r|r|}
    \hline
        $d$ & Total Time & \texttt{Gurobi} Time & \texttt{Gurobi} Iterations & Iterations \\ \hline
        11 & 9.10 & 8.09 & 1 & 8 \\ \hline
        12 & 13.61 & 12.22 & 1 & 5 \\ \hline
        13 & 32.78 & 30.42 & 3 & 16 \\ \hline
        14 & 28.06 & 26.84 & 2 & 15 \\ \hline
        15 & 20.89 & 19.79 & 1 & 10 \\ \hline
        16 & 44.77 & 42.65 & 2 & 24 \\ \hline
        17 & 24.64 & 22.49 & 1 & 23 \\ \hline
        18 & 27.41 & 25.50 & 1 & 16 \\ \hline
        19 & 55.36 & 52.26 & 2 & 29 \\ \hline
        20 & 73.51 & 69.19 & 2 & 35 \\ \hline
    \end{tabular}
    \label{table:local_kn}
\end{table}

Table~\ref{table:colgen-knapsack} shows timings for solving the convex relaxation \eqref{lagrangian_dual_relaxation}. A large portion of the time taken to solve the convex relaxation is spent solving the restricted dual program \eqref{eq:restricted-program} with \texttt{Mosek}. In Figure \ref{figure:colgen-knap}, we show the objective improvements returned by the column generation algorithm for a single instance. When we solve the separation problem \eqref{pricing_LS} optimally, we can use the value to construct a feasible dual solution, as shown in Proposition~\ref{proposition:ip-upper-bound}, to get an upper bound on $\phi_R$ which is called IP bound in the figure. 

\begin{table}[!ht]
    \centering
    \caption{Timing statistics (in seconds) of the column generation algorithm for solving the first-order model with knapsack constraints. Each row corresponds to timings for a single run of the column generation for dimension $d$. Total Time is the total time the algorithm took from beginning to end, \texttt{Gurobi} Time is the total amount of time the algorithm spent in \texttt{Gurobi}, \texttt{Mosek} Time is the total time the algorithm spent in \texttt{Mosek}, and Iterations is the total number of iterations it took.  }
\begin{tabular}{|r|r|r|r|r|r|}
\hline
d  & Total Time & \texttt{Gurobi} Time & \texttt{Mosek} Time & IPs solved & Iterations \\ \hline
11 & 11.45      & 3.44        & 0.13       & 17       & 17         \\ \hline
12 & 20.32      & 8.08        & 1.19       & 33       & 41         \\ \hline
13 & 17.64      & 4.29        & 0.07       & 16       & 18         \\ \hline
14 & 53.99      & 17.36       & 0.11       & 44       & 67         \\ \hline
15 & 113.31     & 32.34       & 43.32      & 82       & 103        \\ \hline
16 & 151.70      & 63.94       & 27.31      & 130      & 201        \\ \hline
17 & 3350.19    & 112.5       & 3173.13    & 181      & 260        \\ \hline
18 & 3392.73    & 118.96      & 3194.11    & 181      & 218        \\ \hline
19 & 22123.85   & 225.19      & 21799.30    & 255      & 340        \\ \hline
20 & 45500.13   & 307.95      & 45083.08   & 236      & 365      \\ \hline 
\end{tabular}
\label{table:colgen-knapsack}
\end{table}

\begin{figure}[!ht] 
\includegraphics[width=17cm]{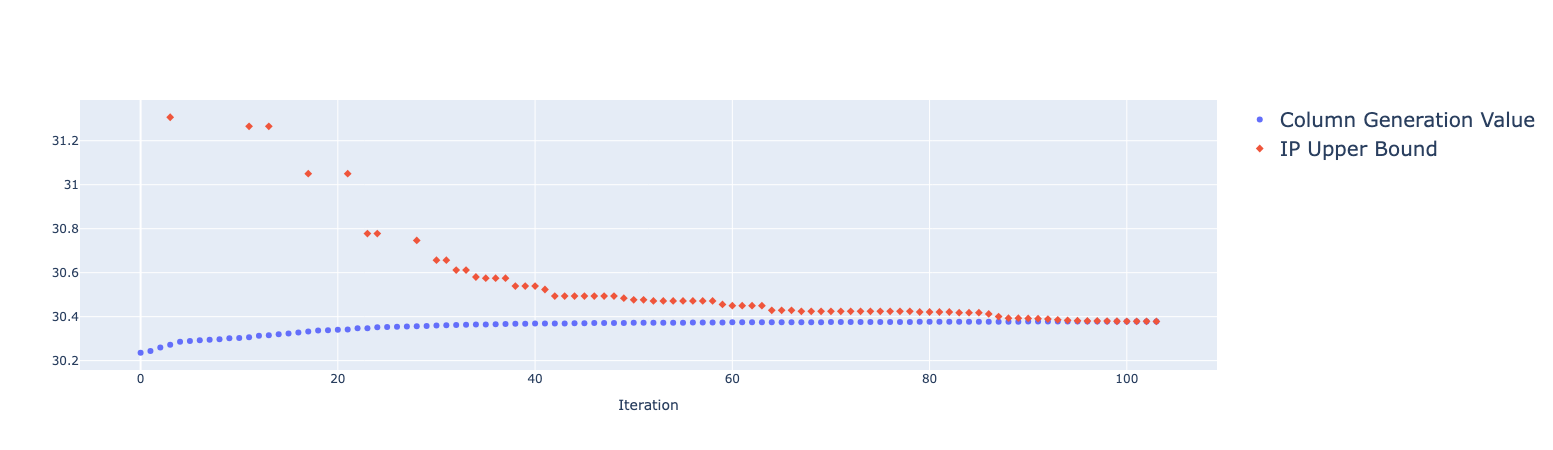}
\centering
\caption{A single run of the column generation algorithm for the first-order model with knapsack constraints for $d=15$. The IP upper bound is the value of a dual feasible solution, and the Column Generation Value is the value of the column generation algorithm.}
\label{figure:colgen-knap}
\end{figure}

\subsection{Second-order model with knapsack constraints}
\begin{figure}[ht]
\includegraphics[width=19cm]{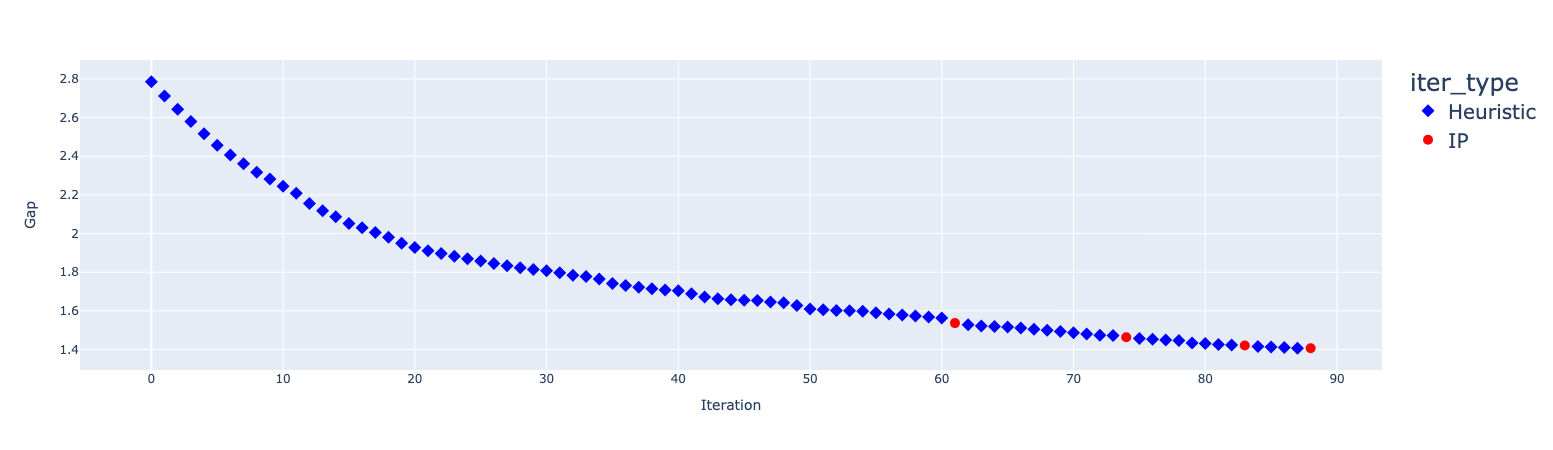}
\centering
\caption{A single run of the local search algorithm for solving second-order model with knapsack constraints for $d=16$. Heuristic means the local search found a locally improving move using the bit flip/swap heuristic and IP means the algorithm failed to find an improving move with the heuristic and had to run \texttt{Gurobi}. Using only the heuristic, the local search would have gotten stuck at iteration $60$ with a gap of $1.562$. Using the IP allows the algorithm to continue by improving the gap to $1.535$ at iteration $61$ and finally terminate at iteration  $88$ with a gap of $1.406$.    }
\label{figure:local-pairs}
\end{figure}

In this experiment, we let the first $d$ monomials as degree one monomials, i.e., we let $m(\bfx) = x_i$ for $i \in [d]$ and the rest are the first $\binom{ \lfloor \frac{d}{2} \rfloor }{2}$ pairs out of all $\binom{d}{2}$ pairs. This means that we also incorporate all the monomials in the form of $ x_\ell x_j$ for all pairs $\{\ell, j\}$ from the set $\{2, 3, \ldots, \lfloor \frac{d}{2} \rfloor + 1 \}$ with $\ell\neq j$. The two knapsack constraints are the same as those used in Section~\ref{sec:experiment-knapsack}.

Figure~\ref{figure:local-pairs} illustrates a single run of the local search algorithm with $d = 20$. We can see that the IP is much more important for finding an improving move for this variant of the design problem. The heuristic failed to find an improving solution during earlier iterations, and the IP was used in more iterations. 

Table~\ref{table:local_pairs} shows the run times of the local search algorithm with a random starting point. The times for this variant are significantly higher than those in the previous variants since the IP was used in each run more frequently, and the dimension of the design points is much larger due to degree-two monomials. 

Table~\ref{table:local_pairs_values} shows a comparison between the value of the solution returned by the local search algorithm and \texttt{JMP}. For this variant, we are able to almost always get a better value than \texttt{JMP} and also improve the \texttt{JMP} solution using the local search. 

\begin{table}[!ht]
    \centering
    \caption{Local search timing statistics for second-order model with knapsack constraints.}
    \begin{tabular}{|r|r|r|r|r|}
    \hline
        $d$ & Total Time & \texttt{Gurobi} Time & \texttt{Gurobi} Iterations & Iterations \\ \hline
        11 & 71.26 & 67.88 & 1 & 9 \\ \hline
        12 & 147.75 & 144.44 & 1 & 30 \\ \hline
        13 & 147.80 & 141.45 & 2 & 30 \\ \hline
        14 & 364.51 & 355.45 & 3 & 38 \\ \hline
        15 & 684.80 & 672.79 & 3 & 63 \\ \hline
        16 &1690.40 & 1664.10 & 4 & 89 \\ \hline
        17 & 1480.17 & 1422.48 & 2 & 97 \\ \hline
        18 & 1373.81 & 1316.34 & 1 & 98 \\ \hline
        19 & 2048.41 & 1976.67 & 2 & 141 \\ \hline
        20 & 6926.83 & 6722.65 & 6 & 247 \\ \hline
    \end{tabular}
    \label{table:local_pairs}
\end{table}

\begin{table}[!ht]
    \centering
    \caption{Local search values for second-order model with knapsack constraints. Then relaxation values listed for $d=19,20$ are the best upper bounds found for the instances using Proposition~\ref{proposition:ip-upper-bound}. For $d=19$, the relaxation value of the optimal solution is in the interval $[132.547, 132.634]$ and for $d=20$ the relaxation value of the optimal solution is in the interval $[164.070, 164.273]$. The lower bound comes is the value of a feasible primal solution, and the upper bound is the value of a feasible dual solution.}
    \begin{tabular}{|r|r|r|r|r|}
    \hline
        $d$ & LS Value & \texttt{JMP} Value & \texttt{JMP} + LS Value & Relaxation Value \\ \hline
        11 & 32.804 & 32.824 & 32.824 & 33.367 \\ \hline
        12 & 46.819 & 46.638 & 46.707 & 47.484 \\ \hline
        13 & 45.761 & 45.505 & 45.627 & 46.426 \\ \hline
        14 & 68.768 & 68.744 & 69.067 & 69.866 \\ \hline
        15 & 73.762 & 73.513 & 73.712 & 74.798 \\ \hline
        16 & 90.827 & 90.272 & 90.681 & 92.233 \\ \hline
        17 & 98.448 & 98.137 & 98.531 & 100.451 \\ \hline
        18 & 120.817 & 119.719 & 120.901 & 122.743 \\ \hline
        19 & 130.339 & 130.183 & 130.384 & $132.634^*$ \\ \hline
        20 & 161.347 & 161.097 & 161.298 & $164.273^*$ \\ \hline
    \end{tabular}
    \label{table:local_pairs_values}
\end{table}

\section{Conclusion}
We developed local search and column generation algorithms for $D$-optimal design problems that are centered around solving a quadratic optimization problem \eqref{pricing_LS}. For the local search algorithm, we established a theoretical approximation guarantee, assuming that \eqref{pricing_LS} can be solved using an algorithm with a known approximation factor. Additionally, we showed that achieving an approximation ratio better than $1/2$ for \eqref{pricing_LS} in the second-order model is NP-Hard. 

We implemented the local search and column generation for  both the first-order and second-order models with knapsack constraints. Our implementation solves \eqref{pricing_LS} with an integer program, and the experimental results in Section~\ref{sec:experiments} highlight the importance of the integer program in obtaining high-quality solutions. As the model complexity increases, basic heuristics such as \eqref{eq:heuristic-1} and \eqref{eq:heuristic-2} fail to find local improvements, making the integer program essential to find improving moves in larger neighborhoods. Notably, our local search starting from a random solution often outperforms \texttt{JMP}, and using the solution returned by \texttt{JMP} as a starting point for local search typically yields further improvements.

In Section~\ref{sec:local-theory}, we established a negative result for the pricing problem \eqref{pricing_LS} in the special case where the matrix $\bm{G}$ is diagonally dominant for the second-order model. A natural theoretical follow-up question is what kind of approximation guarantee can be achieved in the general case, where $\bm{G}$ is any positive semi-definite matrix. Answering this would allow us to apply Theorem~\ref{theorem:approx-local} and obtain a result similar to Corollary~\ref{corollary:full-linear-guarantee} for the second-order model. From an empirical perspective, it would be interesting to use local search and column generation algorithms in a branch-and-bound framework to find an integral optimal solution.

\section*{ Acknowledgments.} A. Pillai and M. Singh were supported in part by NSF grant CCF-2106444 and NSF grant CCF-1910423.
G. Ponte was supported in part by CNPq GM-GD scholarship 161501/2022-2 and AFOSR grant FA9550-22-1-0172.
M. Fampa was supported in part by CNPq grant 307167/2022-4.
J. Lee was supported in part by AFOSR grant FA9550-22-1-0172. W. Xie was supported in part by NSF grant 2246417 and  ONR grant N00014-24-1-2066.

\newpage

\bibliography{Reference}

\end{document}